\documentclass[11pt,a4paper]{article}
\usepackage[dvips]{graphicx}
\usepackage[utf8]{inputenc}
\usepackage{enumerate}
\usepackage[spanish,english]{babel}
\usepackage{amsthm,amssymb,amsfonts,amsmath,esint}
\usepackage[ruled,linesnumbered]{algorithm2e}
\usepackage{lineno}
\usepackage{hyperref}
\usepackage{caption}
\usepackage{subcaption}
\usepackage{xcolor}
\usepackage{cite}
\usepackage{pifont}
\usepackage{enumerate}
\usepackage{appendix}

\newtheorem{theorem}{Theorem}[section]

\newtheorem{proposition}{Proposition}[section]

\newtheorem{lemma}{Lemma}[section]
\newtheorem{corollary}{Corollary}[section]

\newtheorem{property}{Property}[section]

\theoremstyle{remark}
\newtheorem{remark}{Remark}[section]

\graphicspath{{figures/}}

\newcommand{\T}{\ensuremath{{\mathcal{T}}_{v}}\xspace}
\newcommand{\PP}{\ensuremath{{P}_{\ell}}\xspace}
\newcommand{\NL}{\ensuremath{{G}_{\ell}}\xspace}
\newcommand{\N}{\ensuremath{{G}\xspace}}
\newcommand{\W}{W}

\title{ Continuous mean distance of a weighted graph}

\author{\normalfont\fontsize{12}{14}\selectfont
Garijo, Delia$^1$\and
M\'arquez, Alberto$^1$\and
Silveira, Rodrigo I. $^2$
}

\date{}
\begin{document}

\maketitle

\addtocounter{footnote}{1}  \footnotetext{Departamento de Matem\'atica Aplicada I, Universidad de Sevilla, Avda. Reina Mercedes s/n, 41012, Sevilla, Spain. Emails: \{dgarijo, almar\}@us.es.}
\addtocounter{footnote}{1} \footnotetext{Departament de Matem\`{a}tiques, Universitat Polit\`{e}cnica de Catalunya, Jordi Girona 1-3, 08034, Barcelona, Spain. 
Email: rodrigo.silveira@upc.edu.}

\begin{abstract}
We study the concept of the \emph{continuous mean distance} of a  weighted graph.
For connected unweighted graphs, the \emph{mean distance} can be defined as the arithmetic mean of the distances between all pairs of vertices.
 This parameter provides a natural measure of the compactness of the graph, and has been intensively studied, together with several variants, including its version for weighted graphs.
The \emph{continuous} analog of the (discrete) mean distance is the mean of the distances between all pairs of \emph{points} on the edges of the graph. Despite being a very natural generalization, to the best of our knowledge this concept has been barely studied, since 
the jump from discrete to continuous implies having to deal with an infinite number of distances, something that increases the difficulty of the parameter. 

In this paper, we show that the continuous mean distance of a weighted graph can be computed in time roughly quadratic in the number of edges, by two different methods 
that apply fundamental concepts in discrete algorithms and computational geometry. We also present structural results that allow for a faster computation of this continuous parameter for several classes of weighted graphs. Finally, we study the relation between the (discrete) mean distance and its continuous counterpart, mainly focusing on the relevant question of convergence when iteratively subdividing the edges of the weighted graph.
\end{abstract}

\section{Introduction}

Distances are one of the most essential aspects in the analysis of graphs, regardless of whether they originate in geography, transportation, sociology, or communications.
The maximum distance between any two nodes in the graph, known as the \emph{diameter},  provides a worst-case scenario in terms of distances, and gives the maximum \emph{eccentricity} in the graph.
Similarly, the \emph{average} or \emph{mean distance} is related to centrality, and provides a measure of the compactness of the graph.
In this work we will focus on the latter, the mean distance.

The \emph{mean distance} of a connected unweighted graph $\N=(V(\N),E(\N))$ was first introduced by March and Steadman \cite[Chap.14]{MS71} in the context of architecture to compare floor plans, although interest in the concept dates back to the work of Wiener in chemistry~\cite{Wiener47} (after whom the closely related \emph{Wiener index}, the sum of all pairwise distances in the graph,  is named).

The most usual way to define the mean distance $\mu(\N)$ is as the arithmetic mean of all nonzero distances between vertices, i.e., 
\begin{equation}
\label{eq:mean_classic}
   \mu(\N)=\frac{1}{{|V(\N)| \choose 2}}\sum_{\{u,v\}\subseteq V(G)} d(u,v),
\end{equation}
where $|V(G)|\geq 2$, $d(u,v)$ is the length of a shortest path connecting vertices $u$ and $v$, and the sum is taken over all unordered pairs of vertices in the graph.

In the context of graph theory, Doyle and Graver~\cite{DG77} were the first to propose the mean distance as a graph parameter.
Since them, it has been intensively studied.
For very simple graphs, the mean distance is well-understood.
For instance, it is 1 in any complete graph, and $(n+1)/3$ if the graph is an $n$-vertex path.
However, as soon as the graph becomes more complicated, the expression for its mean distance becomes much more elusive.
In addition to presenting exact expressions for a few specific graph classes~\cite{B76,BS81,DG77}, most previous work has focused on proving lower and upper bounds on the mean distance as a function of parameters such as the number of vertices~\cite{DG77,EJS76,P84}, number of vertices and edges~\cite{olts1991TransmissionIG}, and connectivity~\cite{Favaron89}.
Considerable effort was also put into understanding the relation between the mean distance and the minimum vertex degree~\cite{Kouider97}, as well as some spectral graph properties~\cite{Merris89,M91,Rodriguez99}.

The concept of mean distance has also been extended to weighted graphs, both for vertex weights~\cite{D12}, and for rational edge weights~\cite{DG82mds,DG82srmd}.
A few have also studied the concept for directed graphs~\cite{ng1966finite,P84}.
%

As mentioned above, a concept that is closely related to the mean distance is the Wiener index, defined as the sum of distances between all (unordered) pairs of vertices in the graph.
The Wiener index has been studied extensively (both for unweighted and  weighted graphs) due to its important applications in chemical graph theory~\cite{Nikolic95}, but has also received attention in other areas, such as mathematics~\cite{KnorST16} and social graph analysis\cite{Otte02}, and it is still the topic of active investigation (see, e.g.,~\cite{Singh2021}).
From the point of view of computation, there have been important efforts in computer science to understand how efficiently the Wiener index can be computed.
While it is immediate to obtain a roughly quadratic-time algorithm that computes each distance in the graph (e.g., by solving an all-pairs shortest path problem), the challenge is to understand in which situations this can be done more efficiently.
Since for arbitrary graphs it is known that this is not possible unless the strong exponential time hypothesis (SETH) fails~\cite{RodittyW13},
the focus has been on identifying classes of graphs for which the Wiener index can be computed in subquadratic time.
Some cases for which this has been shown to be possible is for graphs with bounded treewidth\cite{CabelloK09,AbboudWW16}, and most notably for planar graphs with non-negative cycles~\cite{Cabello2019}.

Going back to the mean distance, a different direction was adopted by Doyle and Graver~\cite{DG82mds,DG82srmd}, who introduced the mean distance of a \emph{shape}.
This is defined for any weighted graph embedded in the plane.
Each edge of the graph is iteratively subdivided into shorter edges, so that the edge lengths approach zero.
The mean distance of the shape is then defined as the limit of the mean distance of such a sequence of refinements.
While this is a natural definition, its computation is involved.
Doyle and Graver managed  to  compute its exact value for seven specific types of simple  graphs (i.e., a path, a Y-shape, an H-shape, a cross, and three more) and six rather specific families of graphs; the most general ones being  cycles and stars with $k$ edges of length $1/k$.
Examples of the more specific families studied are graphs consisting of one edge with two edges attached at each endpoint, and (multi)graphs consisting of three edges sharing both endpoints, in both cases for very constrained edge lengths.
A summary of these formulas is given in \cite{DG82srmd}; they are obtained as a consequence of the techniques developed in \cite{DG82mds}, mainly, for trees and the so-called \emph{geometric} shapes.

In this paper, we study the mean distance in a \emph{continuous} setting, in a spirit very similar to that of the shapes of Doyle and Graver~\cite{DG82mds,DG82srmd}. Our main motivation arises from \emph{geometric graphs}.
A geometric graph is an undirected graph where each vertex is a two-dimensional point, and each edge is a straight line segment between the corresponding two points. 
Geometric graphs appear naturally in many applications, for instance in road, river or computer graphs.
Unlike abstract graphs, in geometric graphs distances are not only defined for pairs of vertices, but they exist for any two points on the graph, including points on the interior of edges.
Therefore, the concept of mean distance generalizes naturally to (weighted) geometric graphs, defined as the average distance between all \emph{pairs of points} on edges of the graph.
While being a natural definition, the jump from discrete to continuous implies that now the mean is the sum of an infinite number of distances, something that changes the properties of this index and makes its computation difficult.
In this paper, we study this concept in depth, with the focus on the computational aspects of the continuous mean distance, and on understanding how much it differs from the vastly studied discrete mean distance. 

In particular, our main contributions are:
\begin{itemize}
    \item We show that the continuous mean distance of a weigthed graph with $m$ edges can be computed in  $O(m^2)$ time, once all pair-wise distances between vertices have been computed. To this end, we present two different methods, 
    one based on a generalization of shortest path trees to continuous distances, and one based on Voronoi diagrams for the $L_1$ (or Manhattan) metric.
    See Section~\ref{sec:computation}.
    
    \item We present several structural results that allow a faster computation of the continuous mean distance for several classes of weighted graphs.
    In particular, we give an exact expression for complete graphs where all edges have the same length, and efficient algorithms for families of graphs that have a cut vertex, which include weighted trees and weighted cactus graphs.
    See Section~\ref{sec:particular_cases}.
    
    \item We study the relation between the discrete mean distance and the continuous counterpart.
    After establishing some relations between them in Section~\ref{sec:relation}, we move to the relevant question of convergence: When does iteratively subdividing edges and computing the discrete mean distance converge to the continuous mean distance? While a definitive answer to this question does not seem possible, in Section~\ref{sec:convergence_subdivisions} we study a refining procedure that gives a guarantee on how much the discrete and continuous means can differ as the weighted graph is  iteratively refined. The bounds obtained are tight for some graphs classes, such as trees where all edges have the same length.
\end{itemize}

Next we present our problem formally.

\section{Preliminaries}

Let $\N=(V(\N),E(\N))$  be a connected graph\footnote{All graphs considered in this work are assumed to be connected.} with $n$ vertices and $m$ edges; when no confusion may arise, we indistinctly write $V$ or $V(\N)$ and $E$ or $E(\N)$. Consider a function $\omega: E\longrightarrow \mathbb{R}^+$ that assigns a positive weight $\omega(e)$ to each edge $e \in E$. The value $\omega(e)$ is called the \emph{length} of edge $e$, and is also denoted by $|e|$. In general, given a subset of edges $E'\subseteq E$, its \emph{weight} or \emph{length} is $|E'|=\sum_{e\in E'} \omega(e)$.

Graph $\N$ together with function $\omega$ is a weighted graph where every edge can be identified with a line segment of length $\omega(e)$ in the Euclidean plane.
 Thus, every point $p$ on an edge $e=uv$ can be expressed as $p=\lambda_p v+(1-\lambda_p)u$ for some $\lambda_p \in [0,1]$. Let $\NL$ be the set of all points that are on the edges on $\N$.
 Note that this definition not only includes all geometric graphs, but  also covers other graphs that are not geometric. A simple example of such a graph is a triangle where two edges have length 1 and the third one has length 2; such a graph cannot be realized with three straight line segments, since it would require the longer edge to overlap with the two shorter ones. 
 
 We point out that all graphs considered in this work are connected and weighted, although both terms will be in general omitted as it is understood from the context. We will also consider \emph{uniform} graphs: graphs where all edges have the same length. We will write \emph{$\alpha$-uniform} to refer to a uniform graph where all edge lengths are $\alpha$.

Let $p, q$ be two points on $\NL$ that are not both  on the interior of the same edge. A \emph{path} $\PP$ between $p$ and $q$, also called \emph{$pq$-path}, is a sequence $pu_1\dots u_kq$ where 
and the \emph{distance} $d(p,q)$ between $p$ and $q$ on ${\N}_{\ell}$ is the length of a shortest path connecting the two points. When the two points $p,q$ are on the interior of the same edge $u_0u_1$ and, say $\lambda_p<\lambda_q$, we have paths between $p$ and $q$ that go through vertices (whose definition is analogous to the above one) but also a path in the interior of the edge that is the segment connecting $p$ and $q$ (edges are identified with segments), and its length is $(\lambda_q-\lambda_p)\omega(u_0u_{1})$.
In this paper, we shall assume that the distance between the two endpoints of any edge $e$ is $|e|$.
The set of points $\NL$ together with this distance function is a metric space, and it will be treated indistinctly as a graph (with vertex set $V(\NL)=V(\N)$ and edge set $E(\NL)=E(\N)$) or as a point set. The distance between an edge $e=uv$ and a point $p\notin e$ is $d(p,e)={\rm min}\{d(p,u),d(p,v)\}$ (if $p \in e$, $d(p,e)=0$), and the distance between two edges $e$ and $e'=ab$ is $d(e,e')={\rm min}\{d(a,e),d(b,e)\}$. 

We begin by defining the variant of the \emph{discrete mean distance} that we will consider in the remainder of this work.
The definition below differs from the Equation~(\ref{eq:mean_classic}) of $\mu(\N)$  in two aspects: 
(i) it considers \emph{all} pairs of distances, including those that are zero, and
(ii) it considers \emph{ordered} pairs of vertices:
\begin{equation}\label{eq:discrete}
\mu_d(\N)=\frac{1}{n^2}\displaystyle\sum_{(u,v) \in V\times V} d(u,v) = \frac{2\W(\N)}{n^2},    
\end{equation}
where $\W(\N)$ denotes the Wiener index of $\N$.
Observe that $\mu_d(\N)$ is the arithmetic mean of the entries of the distance matrix of the graph.
Although this alternative form of mean distance has been considered before~\cite{mean_dist_M}, our motivation for studying it comes from the fact that it extends better to the continuous mean distance (which is the subject of this paper) in a limiting process when iteratively subdividing the edges of the graph.
In particular, it will allow us to establish a clear relation between the discrete and the continuous mean distance.

To define formally the continuous mean distance of a weighted graph, we start by defining it between a point and a set of edges.
Given a point $p \in \NL$ and a subset of edges  $E'\subseteq  E(\NL)$, the \emph{continuous mean distance} between $p$ and $E'$ is
\begin{equation}\label{eq:def0}
{\mu_c(p,E')}={\frac {1}{|E'|}}\int_{q \in E'} d(p,q)\, dq.
\end{equation}
For subsets of edges $E', E'' \subseteq E(\NL)$, the \emph{continuous mean distance} between $E'$ and $E''$ is
\begin{equation}\label{eq:def1}
\mu_c(E',E'')
  =
  \frac {1}{|E'||E''|} \iint_{p\in E', \, q \in E''} d(p,q)\,dp\,dq.
\end{equation}
With some abuse of notation, we shall write $\mu_c(p,\N')$ or $\mu_c(\N',\N'')$, where $\N'$ and $\N''$ are the graphs with edge sets $E'$ and $E''$, respectively. 

Based on the previous, the \emph{continuous mean distance} of the weighted graph $\NL$ is defined as 
\begin{equation}\label{eq:def2}
\mu_c(\NL)=\mu_c(E(\NL),E(\NL)).
\end{equation}

\begin{remark} \label{rem:homotecia}
In \cite[Equations (1) and (3)]{DG82mds}, the authors show that the discrete mean distance of an $\alpha$-uniform graph $\N$ can be easily deduced from the $1$-uniform case.
They use definition (\ref{eq:mean_classic}) and $\alpha \in \mathbb{Q}^+$: $\mu(\N)=\alpha \mu(\N_1)$ where $\N_1$ is the corresponding $1$-uniform graph. This can be naturally extended to the variant $\mu_d(\N)$ and $\alpha \in \mathbb{R}^+$,
and by elementary properties of integration, a similar formula holds for the continuous mean distance   even when the graph $\NL$ is not uniform: $\mu_c(\NL')=\beta \mu_c(\NL)$ where $\NL'$ is the graph obtained by simply multiplying all edge lengths of $\NL$ by $\beta$.
\end{remark}

The following observations, which follow directly from Equation~(\ref{eq:def1}), will be used throughout this work.


\begin{remark} \label{rem:division_edge} 
Let $ab$ and $uv$ be two edges in $E(\NL)$. 
\begin{itemize}
\item[(i)] For any point $p \in uv$ we have
\[ 
\mu_c(uv, ab) = \frac{|up|\mu_c(up, ab) + |pv|\mu_c(pv, ab)}{|uv|}.
\]
\item[(ii)] For each point $p\in uv$ and each point $q\in ab$,
if $d(p, q) = d(p, v)+d(v, q)$,
we have
\[\mu_c(uv, ab) = \mu_c(uv, v)+\mu_c(v, ab).\]
\end{itemize}
\end{remark}


\paragraph{An example: paths}
It is illustrative to see how the continuous mean distance can differ from the discrete version.
Here we illustrate this for the important case of paths.

Consider a $1$-uniform path $P$, i.e., a graph consisting of a path with $n$ vertices and all edges of length 1. The discrete mean distance of such a path is known to be $\mu_d(P)=(n^2-1)/3n$~\cite{ mean_dist_M}. By Remark \ref{rem:homotecia}, this generalizes to $\mu_d(P)=\alpha (n^2-1)/3n$ when $P$ is $\alpha$-uniform (its total length is $\alpha(n-1)$). For non-uniform paths $P$, there is no closed formula to compute $\mu_d(P)$.
In contrast, it is possible to obtain a closed formula for $\mu_c(\PP)$, for any path with arbitrary positive real edge lengths, as explained next.
 
First observe that, for the continuous mean distance, the number of interior nodes in a path does not play any role, thus we can consider the path as one single edge.
Hence a path $\PP$ of length $t \in \mathbb{R}_{>0}$ can be seen as the interval $[0,t]$. For a point $x\in [0,t]$, let $d(x,[0,t])$ denote the function that gives the distance between $x$ and any other point $x'$ in the interval $[0,t]$; the shape of this function is illustrated in Figure~\ref{fig:interval}.
The mean value of $d(x,[0,t])$ is $\frac{1}{2t} (x^2+(x-t)^2)$.\footnote{Recall that the mean value of a function $f$ over an interval $[a,b]$ is
${\frac {1}{b-a}}\int _{a}^{b}f(x)\,dx.$} Thus,
\begin{equation}
\label{eq:path_formula}
\mu_c(\PP)=\frac{1}{t} \int_{0}^{t} \frac{1}{2t} (x^2+(x-t)^2)\, dx=\frac{t}{3}.
\end{equation}
Based on a different approach, the same value was given in~\cite{DG82mds} (see also \cite{DG82srmd}) for $\alpha$-uniform paths $\PP$ with $\alpha\in \mathbb{Q^+}$.

\begin{figure}[ht]
\centering
\includegraphics{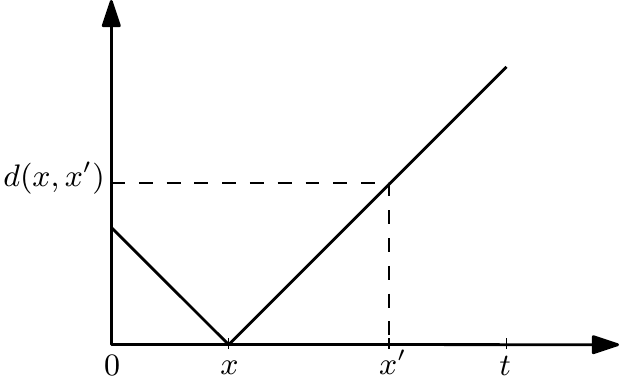}
\caption{Function $d(x,[0,t])$.}
\label{fig:interval}
\end{figure}

\section{Computation of the continuous mean distance}
\label{sec:computation}

The continuous nature of the continuous mean distance makes its computation non-trivial, as exemplified by the seemingly simple case of paths. 
In this section, we show that despite this,  $\mu_c(\NL)$  can be computed rather efficiently, in time roughly quadratic in the number of edges of $\NL$.
We will show how this can be achieved in two different ways, which apply some fundamental concepts in discrete algorithms and computational geometry: that of shortest path trees and that of Voronoi diagrams for the $L_1$ (or Manhattan) metric.
We highlight that the relation between the continuous mean distance and these two ubiquitous structures is interesting on its own. 

The main result of this section is the following.

\begin{theorem} \label{th:quadratic}
 The continuous mean distance of a  weighted graph $\NL$ with $n$ vertices and $m$ edges can be computed in $O(m^2+A(n,m))$ time, where $A(n,m)$ is the  time required to compute all vertex-to-vertex distances in \NL.
\end{theorem}

To prove the preceding theorem, we use the following formula, which states that $\mu_c(\NL)$ can be obtained as a weighted sum of the continuous mean distances of all \emph{ordered} pairs of edges; this is simply a consequence of Equations~(\ref{eq:def1})--(\ref{eq:path_formula}), and elementary properties of integration. 

\begin{equation} \label{eq:sum}
\mu_c(\NL) = \frac{1}{|E|^2}\left( \sum_{(e,e')\in E \times E, e \neq e'}{\mu_c(e,e')  |e||e'| + \sum_{e \in E}{\frac{|e|}{3}  |e|^2 }} \right).
\end{equation}

This fact reflects that understanding how the continuous mean distance behaves in the case of two edges is the key tool to compute it for the whole graph. In the next subsections, we present our two different approaches for the two-edge case. Theorem  \ref{th:spt} as well as Theorem \ref{cor:floor} let us conclude that the continuous mean distance between two edges can be computed in constant time, once the distance matrix of the vertices of the graph $\N$ has been computed.
The currently best algorithm to compute all-pairs shortest paths for a graph with real weights has running time $A(n,m)=O(nm \log \alpha(m, n))$~\cite{pr-sparwug-05}, where $\alpha(m,n)$ is the extremely slowly growing inverse of the Ackermann function.
We note that for some special graph classes faster algorithms are known, such as planar graphs with non-negative edge weights (where $A(m,n)=O(n^2)$~\cite{HENZINGER1997}), or graphs with integer non-negative edge weights (for which $A=O(nm)$~\cite{T-99}).


The continuous mean distance of two equal edges reduces to the mean distance of a path, which---as we have seen in Equation~(\ref{eq:path_formula})---is equal to  $\mu_c(e,e)=|e|/3$, for any edge $e$; this is used in Equation~(\ref{eq:sum}). Therefore, in the remainder of the section, we focus on the mean distance between two distinct edges.





\subsection{Computation using shortest path trees}
\label{subsec:trees}

Shortest path trees are one of the most fundamental structures used to represent distances in graphs, and they are an essential underlying concept behind most single-source shortest path algorithms.
In this section, we introduce a continuous version of the shortest path tree rooted at a vertex of a weighted graph $\N$, and later show how it can be used to compute the continuous mean distance between any two distinct edges of $\NL$.

For $\NL$ and a vertex $v\in V(\NL)$, a \emph{continuous shortest path tree} is a pair $\T=(T_v,S_v)$, where $T_v$ is a (discrete) shortest path tree rooted at $v$, and $S_v$ is a subset of $\NL$ that contains one point $p_e^v$ for each edge $e\in E(\NL)\backslash E(T_v)$. Point $p_e^v$ is the only point on edge $e=ab$ such that its distance to $v$ is given by two different paths: one passes through $a$, and the other one passes through $b$. Thus, $p_e^v$ is the furthest point to $v$ on any cycle $C_e^v$ determined by $e$ and shortest paths connecting $v$ with $a$ and $b$ (see Figure~\ref{fig:lem1}). Note that $p_e^v$ must exist, otherwise $e\in E(T_v)$. 
Observe also that for any point on $ap_e^v$, its shortest paths to $v$ go through $a$, and analogously for points on $p_e^vb$.

\begin{figure}[ht]
\centering
\includegraphics{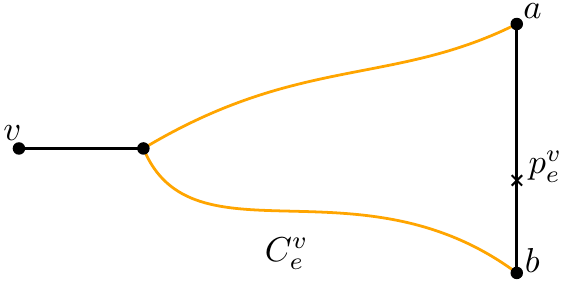}
\caption{Point $p_e^v$, for $e=ab$, is the point furthest from $v$ on cycle $C_e^v$ (composed of the orange paths and edge $ab$).}
\label{fig:lem1}
\end{figure}

As we show next, the continuous shortest path tree can be computed within the same running time needed to solve the single-source shortest path problem, denoted by $S(n,m)$.
Currently, we have  $ S(n,m) = O(m \log \alpha(m,n))$ time in general~\cite{pr-sparwug-05}, in $S(n,m)=O(n)$ time for planar graphs with non-negative edge weights~\cite{HENZINGER1997}, and in $S(n,m)=O(m)$ time for graphs with integer non-negative edge weights~\cite{T-99}.

\begin{proposition}\label{prop:cspt-construction}
Let $\N=(V,E)$ be a weighted graph with $n$ vertices and $m$ edges, and let $v\in V$. A continuous shortest path tree $\T=(T_v,S_v)$ of $\NL$ can be computed in $O(S(n,m))$ time, where $S(n,m)$ is the time required to compute a shortest path tree from $v$.
\end{proposition}

\begin{proof}
We first compute $T_v$ using a single-source shortest path algorithm in $O(S(n,m))$ time.
Now, let $e=ab \in E\backslash E(T_v)$ and $p_e^v\in S_v$. 
Since $p_e^v=\lambda_e b+(1-\lambda_e)a$ for some $\lambda_e\in[0,1]$, and $p_e^v$ is the furthest point from $v$ on any cycle $C_e^v$, we have: 
$$\lambda_e|ab|+d(v,a)=d(v, p_e^v)=\frac{|ab|+d(v,b)+d(v,a)}{2}\Rightarrow \lambda_e=\frac{|ab|+d(v,b)-d(v,a)}{2|ab|}.$$
Thus, $\lambda_e$ can be computed in constant time. Therefore, the set $S_v$ can be computed in total $O(m)$ time, once the single-source shortest path tree is available.
\end{proof}

Depending on whether an edge $e$ belongs to $E(T_v)$, Lemma \ref{lem:vertex} below provides a different expression for its continuous mean distance to vertex $v$. In the proof of this lemma, and throughout this paper, we shall use that, by the mean value theorem for integrals, the mean value of a function $f$ over an interval $[a,b]$ coincides with the height of the rectangle with base $b-a$ and area $\int _{a}^{b}f(x)\,dx$.

\begin{lemma} \label{lem:vertex}
Let $v\in V(\NL)$ and $e=ab \in E(\NL)$. Then,
$$
\mu_c(v,e)=
\begin{cases}
\displaystyle\min \{d(v,a), d(v,b)\} +\frac{|ab|}{2}
&
\text{if } e \in E(T_v),
\\ \noalign{\medskip}
\displaystyle\left(
  d(v,a) +\frac{|ap_e^v|}{2}
\right)
\lambda_e^v+
\left(
  d(v,b) +\frac{|p_e^vb|}{2}
\right)(1- \lambda_e^v)
&
\text{if } e\notin E(T_v),
\end{cases}
$$
where $\displaystyle{|ap_e^v|=\frac{|ab|+d(v,b)-d(v,a)}{2}}$ and $\displaystyle{|p_e^vb|=\frac{|ab|+d(v,a)-d(v,b)}{2}}.$
\end{lemma}

\begin{proof}
Suppose first that $e\in E(T_v)$. Figure~\ref{fig:lemma4}a illustrates the graph of $d(x,v)$ for $x\in e$ assuming that $d(a,v)<d(b,v)$ (analogous otherwise). This is a straight-line segment with slope 1, so the height of the rectangle with base $b-a$ and area $\int _{a}^{b} d(x,v) \,dx$ is $d(a,v)+|ab|/2$. Hence, the result follows.

Assume now that $e\notin E(T_v)$. We can argue as above but considering two rectangles determined by the function $d(x, v)$, one for $x\in [a, p_e^v]$ and the other for $[p_e^v, b]$, where $p_e^v=\lambda_e b+(1-\lambda_e)a$ for some $\lambda_e\in[0,1]$; see Figure~\ref{fig:lemma4}b. The heights of these rectangles are, respectively,  $d(v,a) +\tfrac12|ap_e^v|$ and $d(v,b) +\tfrac12|p_e^vb|$. To obtain $\mu_c(v,e)$ each of these values must be multiplied by the proportion of the segment that corresponds to the base. 

The expressions for $|ap_e^v|$ and $|p_e^vb|$ come from the fact that $p_e^v$ is the farthest point from $v$ on any cycle $C_e^v$ (see Figure~\ref{fig:lem1}). Thus, 
$$d(v,a)+|ap_e^v|=\frac{|ab|+d(v,b)+d(v,a)}{2}=d(v,b)+|bp_e^v|.$$
\end{proof}

\begin{figure}[ht]
\begin{center}
\includegraphics{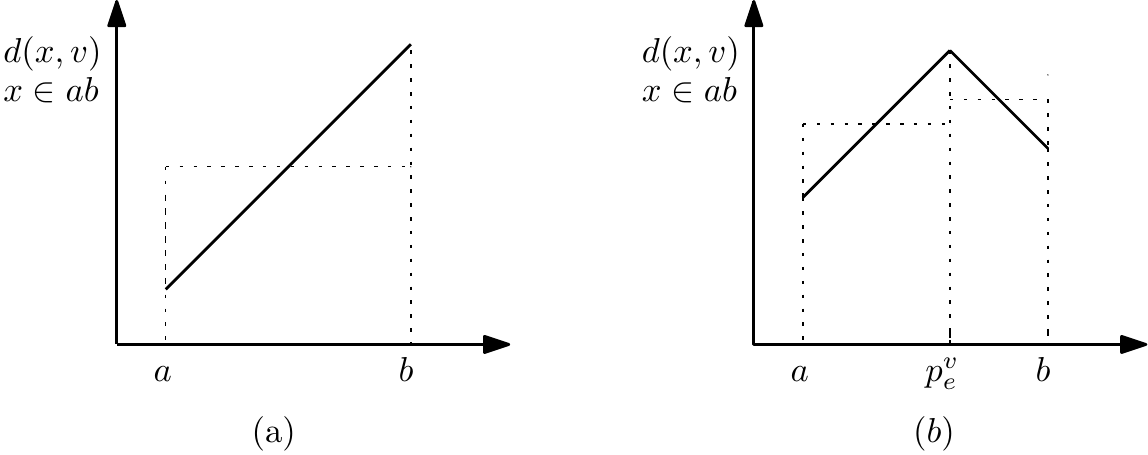}
\end{center}
\caption{Function $y=d(x,v)$ when: (a) $e\in E(T_v)$, (b) $e\notin E(T_v)$.}\label{fig:lemma4}
\end{figure}

The preceding lemma will be used to compute the continuous mean distance between two edges. To do this, we distinguish three cases (Lemmas \ref{lem:star_case}--\ref{lem:graph_case2} below), two of which depend on the following property.

\begin{property}[Same component property]\label{scp}
Let $ab\in E(\NL)$. 
 A vertex $v\in V(\NL)\backslash \{a,b\}$ satisfies the \emph{same component property} with respect to edge $ab$ and vertex $a$ if the shortest path from $a$ to $v$ goes through $b$.
\end{property}

\noindent The same component property essentially means that all  shortest paths from points on edge $ab$ to $v$ go through $b$. 

\begin{lemma}[Rectangular case]
\label{lem:star_case}
Let $ab, uv \in E(\NL)$ be two distinct edges such that:
\begin{enumerate}
\item[(i)] $|ab|=|uv|=\lambda$, 
\item[(ii)] $d(a,uv)$ is given by an $au$-path of length $\theta$, and $d(b,uv)$ is given by a $bv$-path of the same length,
\item[(iii)] the paths of (ii) do not intersect.
\end{enumerate}
Then  $\displaystyle\mu_c(ab,uv)=\theta + \frac{2\lambda}{3}$.

\end{lemma}

\begin{proof}
Edges $ab$ and $uv$ can be seen as the interval $[0,\lambda]$. By Equation (\ref{eq:def1}), we have:
$$\mu_c(ab,uv)= \frac{1}{\lambda^2}\int_{0}^{\lambda}\int_{0}^{\lambda} d(x,y) \,dx\,dy = \frac{1}{\lambda}\int_{0}^{\lambda} \mu_c(x,uv) \,dx$$
Given a point $x\in [0,\lambda]$ on edge $ab$, its furthest point on the cycle composed by $ab$, $uv$, the $au$-path of length $\theta$, and the $bv$-path of the same length is the point $\lambda-x$ on edge $uv$. This point plays the role of $p_e^v$ in Lemma \ref{lem:vertex} with $\lambda_e^v=(\lambda-x)/\lambda$. Note that  Lemma \ref{lem:vertex} is stated for vertices and edges of the graph, but we can always insert a vertex at the required point $x$ and consider, with some abuse of notation, the tree $T_x$. Since the shortest $x-u$ and $x-v$ paths do not contain the edge $uv$, we have $uv\notin E(T_x)$. Then, by Lemma \ref{lem:vertex},
$$\mu_c(x,uv)=\left(x+\theta+\frac{\lambda-x}{2}\right)\frac{\lambda-x}{\lambda}+\left(\lambda-x+\theta+\frac{x}{2}\right)\frac{x}{\lambda},$$
where $d(x,u)=x+\theta$ and $d(x,v)=\lambda-x+\theta$. Hence,
$$\mu_c(ab,uv)=\frac {1}{\lambda^2}\int _{0}^{\lambda} \left[\left(x+\theta+\frac{\lambda-x}{2}\right)(\lambda-x) +\left(\lambda-x+\theta+\frac{x}{2}\right)x\right]\,dx.$$
Thus, $$\mu_c(ab,uv)=\frac {1}{\lambda^2}\int _{0}^{\lambda} (-x^2+\lambda x+\lambda\theta+\frac{\lambda^2}{2}) \, dx=\theta + \frac{2\lambda}{3}.$$
\end{proof}

\begin{lemma} [Linear case] \label{lem:graph_case1}
Let $ab \in E(\NL)$ be an edge such that $a$ and $b$ satisfy the same component property with respect to other edge $uv$ and one of its endpoints, say $u$.
 Then,
    $$
    \mu_c(uv,ab)
    =
    \frac{|uv|}{2}+\mu_c(v,ab).
    $$
\end{lemma}

\begin{proof}

In this case, there are two possible situations that may happen for $ab$ and $uv$, see Figure~\ref{fig:linear_case}.
For each $p \in uv$ and each $q \in ab$, we have $d(p, q) = d(p, v)+d(v, q)$. By  Remark \ref{rem:division_edge}(ii), it follows that $\mu_c(uv, ab) = \mu_c(uv, v)+\mu_c(v, ab)=\frac{|uv|}{2}+\mu_c(v,ab)$. 
\end{proof}

\begin{figure}[t]
\centering
\includegraphics{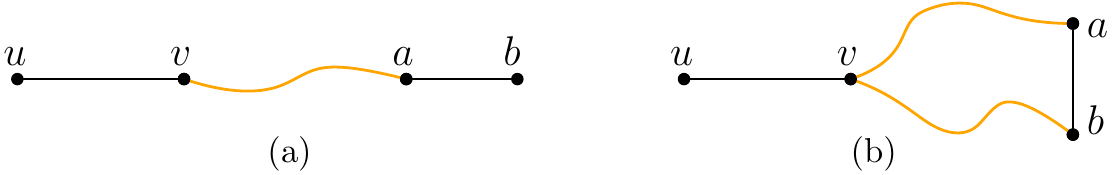}
\caption{Linear cases: (a) $ab\in E(T_u)$, (b) $ab\notin E(T_u)$.}
\label{fig:linear_case}
\end{figure}

\begin{lemma} [Cycle case] \label{lem:graph_case2}
Let $ab, uv \in E(\NL)$ be two distinct edges such that neither
$u,v$ nor
$a,b$ satisfy the same component property with respect to the other corresponding edge and one of its endpoints. Then, $\mu_c(ab,uv)$ can be computed as a weighted sum of at most four linear cases and one rectangular case.
\end{lemma}

\begin{proof}
Consider the continuous shortest path trees $(T_a, S_a)$ and $(T_b, S_b)$ rooted at $a$ and $b$, respectively. Figure~\ref{fig:cycle_case}a illustrates how $T_a$ and $T_b$ must be located with respect to edges $ab$ and $uv$, since no endpoint of these two edges satisfy the same component property. Note that $T_a$ and $T_b$ might have paths in common. Observe also that $uv\notin E(T_a)\cap E(T_b)$.

Let $p_{uv}^a\in S_a$ and $p_{uv}^b\in S_b$. Suppose first that the two points are distinct and, assume without loss of generality, that $p_{uv}^a$ is closer to $u$ than $p_{uv}^b$, see Figure~\ref{fig:cycle_case}a. 
Applying Remark~\ref{rem:division_edge}(i) twice, we have $\mu_c(ab,uv)=\frac{
|up_{uv}^a|}{|uv|}\mu_c(ab,up_{uv}^a)+
\frac{
|p_{uv}^a p_{uv}^b|}{|uv|}\mu_c(ab,p_{uv}^ap_{uv}^b)+
\frac{|p_{uv}^b v|}{|uv|}\mu_c(ab,p_{uv}^bv)$.
Further, $\mu_c(ab,up_{uv}^a)$ can be computed by Lemma \ref{lem:graph_case1} as a linear case: it suffices to consider the shortest paths in $T_a$ and $T_b$ connecting, respectively, $a$ and $b$ with $u$; vertices $a$ and $b$ satisfy the same component property with respect to $p_{uv}^au$ and point $p_{uv}^a$. The situation is analogous for $\mu_c(ab,p_{uv}^bv)$, and it remains to obtain $\mu_c(ab,p_{uv}^ap_{uv}^b)$. Note that Property \ref{scp} is stated for vertices and edges of the graph, but one can always insert vertices at the required points, such as $p_{uv}^a$, in order to deal with the situation as a linear case.  

We now consider the cycle $\mathcal{C}$ determined by edges $ab$ and $uv$, and the shortest paths in $T_a$ and $T_b$ giving, respectively, $d(a,uv)$ and $d(b,uv)$. Suppose, without loss of generality that those paths, denoted by $P_{av}$ and $P_{bu}$,  connect $a$ with $v$ and $b$ with $u$ (see Figure~\ref{fig:cycle_case}b). Observe that $\mathcal{C}$ is indeed a cycle as $p_{uv}^a\neq p_{uv}^b$; otherwise we would have a common sub-path in $P_{av}$ and $P_{bu}$, so there would be a point satisfying that its further point on $uv$ would be the same as the furthest points of $a$ and $b$ on $uv$, which would imply $p_{uv}^a= p_{uv}^b$. Note also that $C$ is a cycle of minimum length containing edges $ab$ and $uv$ since $|P_{av}|=d(a,uv)$ and $|P_{bu}|=d(b,uv)$.


Let $p_{uv}^{a*}$ be the point furthest from $p_{uv}^a$ on $\mathcal{C}$, and let $p_{uv}^{b*}$ be defined analogously for $p_{uv}^b$.
Refer to Figure~\ref{fig:cycle_case}b. By construction of the cycle, these points are on edge $ab$, since $\mathcal{C}$ contains a shortest $ap_{uv}^a$-path and a shortest $bp_{uv}^b$-path. This implies that the furthest point of $a$ on $\mathcal{C}$ is either $p_{uv}^a$ or is located in between $p_{uv}^a$ and $b$. Analogously,  the furthest point of $b$ on $\mathcal{C}$ is either $p_{uv}^b$ or is located in between $p_{uv}^b$ and $a$.  As the order of furthest points has to be preserved in every cycle, $p_{uv}^{a*}$ and $p_{uv}^{b*}$ have to be located in between $a$ and $b$.

Applying Remark~\ref{rem:division_edge}(i) twice from ``edge'' $p_{uv}^ap_{uv}^b$ to $ab$, with points $p_{uv}^{a*}$ and $p_{uv}^{b*}$ in $ab$, yields
$$\mu_c(ab,p_{uv}^ap_{uv}^b)=
\frac{|ap_{uv}^{a*}|}{|ab|}
\mu_c(ap_{uv}^{a*}, p_{uv}^ap_{uv}^b)+
\frac{|p_{uv}^{a*}p_{uv}^{b*}|}{|ab|}
\mu_c(p_{uv}^{a*}p_{uv}^{b*}, p_{uv}^ap_{uv}^b)+
\frac{|p_{uv}^{b*}b|}{|ab|}
\mu_c(p_{uv}^{b*}b, p_{uv}^ap_{uv}^b).$$ 

\begin{figure}[t]
\centering
\includegraphics{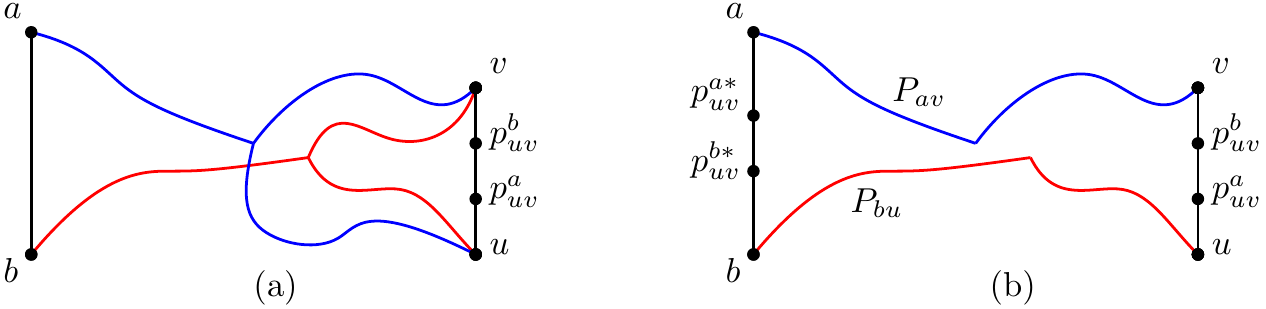}
\caption{ (a) Trees $T_a$ (in blue) and $T_b$ (in red), and points $p_{uv}^a$ and $p_{uv}^b$. (b) Cycle $\mathcal{C}$ formed by the two edges and the paths $P_{av}$ and $P_{bu}$; the points $p_{uv}^{a*}$ and $p_{uv}^{b*}$ are the furthest points of, respectively, $p_{uv}^a$ and $p_{uv}^b$ on $\mathcal{C}$.}
\label{fig:cycle_case}
\end{figure}

With an analogous argument as above, $\mu_c(ap_{uv}^{a*}, p_{uv}^ap_{uv}^b)$ and $\mu_c(p_{uv}^{b*}b, p_{uv}^ap_{uv}^b)$ can be obtained by Lemma \ref{lem:graph_case1} as linear cases (for the first case, for example, take the two shortest paths connecting, respectively, $p_{uv}^a$ and $p_{uv}^b$ with $a$, which go through $u$ and $v$; $p_{uv}^a$ and $p_{uv}^b$ satisfy the same component property with respect to $ap_{uv}^{a*}$ and endpoint $p_{uv}^{a*}$). 

The value $\mu_c(p_{uv}^{a*}p_{uv}^{b*}, p_{uv}^ap_{uv}^b)$ can be computed as a rectangular case of Lemma \ref{lem:star_case}.
 It is easy to check that $|p_{uv}^ap_{uv}^b|=|p_{uv}^{a*}p_{uv}^{b*}|$ as the distances $d(p_{uv}^a,p_{uv}^{a*})$ and $d(p_{uv}^b,p_{uv}^{b*})$ equal the semiperimeter of $\mathcal{C}$. This also implies that the paths on $\mathcal{C}$ connecting, respectively, $p_{uv}^{b*}$ with $p_{uv}^a$, and $p_{uv}^{a*}$ with $p_{uv}^b$, have the same length. 
 
In total, there are at most four linear cases and one rectangular case in order to obtain $\mu_c(ab,uv)$. 
If $p_{uv}^a$ and $p_{uv}^b$ are the same point, the number of linear cases reduces to two and there is no rectangular case since $\mu_c(ab,uv)=\frac{
|up_{uv}^a|}{|uv|}\mu_c(ab,up_{uv}^a)+
\frac{|p_{uv}^a v|}{|uv|}\mu_c(ab,p_{uv}^av)$.
 \end{proof}

Next, we observe that the conditions that need to be checked to compute the continuous mean distance between two edges can be checked in constant time. 
Notice that we do not need the explicit construction of the continuous shortest path trees from each vertex, we only need to check the conditions in Lemmas \ref{lem:condition1} and \ref{lem:condition2} below.  

\begin{lemma}\label{lem:condition1}
\label{lem:edgeInTree}
Given $v\in V(\NL)$,  $e=ab\in E(\NL)$, and the values of $d(v,b)$ and $d(v,a)$,  it can be checked in constant time whether edge $e$ belongs to $E(T_v)$.
\end{lemma}
\begin{proof}
It follows from the fact that $e \notin E(T_v)$ if and only if $|d(v,b)-d(v,a)|<|ab|$.
\end{proof}

\begin{lemma}\label{lem:condition2}
\label{lem:InComponent}
For every edge $uv\in E(\NL)$, it can be checked in constant time whether $u$ and $v$ satisfy the same component property with respect to any other edge  $ab\in E(\NL)$ and one of its endpoints, assuming that $d(a,u), d(a,v), d(b,u)$, and $d(b,v)$ are known.
\end{lemma}
\begin{proof}
Consider an edge $ab$ and the endpoint $a$. Having the same component property with respect to edge $ab$ and vertex $a$ is equivalent to say that
(i) $d(a,u)=d(a,b) + d(b,u)$, and
(ii) $d(a,v)=d(a,b)+d(b,v)$.
\end{proof}


Since any pair of distinct edges falls into one of the three cases considered above (rectangular, linear, or cycle), we conclude the following.

\begin{theorem}\label{th:spt}
Let $G$ be a  weighted graph. Given two edges $e,e' \in E(\NL)$, the function $\mu_c(e,e')$ can be expressed as a weighted sum of $O(1)$ distances between pairs of points on $e$ and $e'$.
\end{theorem}

\subsection{A geometric view based on lower envelopes and Voronoi diagrams} \label{subsec:roof}

In this subsection, we present an alternative approach based on well-known geometric tools, which shows that the continuous mean distance between two edges and, therefore, of the whole graph can be computed completely using simple geometric arguments.
Recall that the \emph{lower envelope} of a set of functions is the function resulting from taking the point-wise minimum of all functions in the set.

\begin{lemma}
\label{lem:two_edges}
Given two distinct edges $e,e'\in E(\NL)$, the function $d(p,q)$, where $p \in e$ and $q \in e'$ can be seen as the lower envelope of at most four planes in 3D.
\end{lemma}
\begin{proof}
Any path connecting points $p$ and $q$ must go through an endpoint of $e=uv$ and an endpoint of $e'=u'v'$. Assume first that the four endpoints are distinct.  Parametrizing the points on $e$ and $e'$ as $p=xv+(1-x)u$ and $q=yv'+(1-y)u'$ for $x,y\in [0,1]$, 
we obtain four planes. For each of the four possible pairs of endpoints of $e$ and $e'$, the corresponding plane gives the length of a shortest $pq$-path among the $pq$-paths that go through those endpoints; their equations are:  
\begin{equation} \label{planes}
    \begin{array}{cl}
   P(u,u'): &     z=|e|x+|e'|y+d(u,u') \\
    P(u,v'): &     z=|e|x+|e'|(1-y)+d(u,v')  \\
     P(v,u'): &    z=|e|(1-x)+|e'|y+d(v,u')\\
    P(v,v'): &     z=|e|(1-x)+|e'|(1-y)+d(v,v')
    \end{array}
\end{equation}
where, for instance, $P(u,u')$ indicates that the paths considered to connect the points on $e$ with those on $e'$ go through endpoints $u$ and $u'$. Hence, for any two points $p\in e$ and $q\in e'$, the function $d(p,q)$ is the minimum among the four values obtained. 

The above argument can be adapted naturally when $e$ and $e'$ have a common endpoint, in which case there are only two planes.
\end{proof}

The previous result implies that the continuous mean distance between any two edges can be computed in constant time if the distances between their endpoints are known.
However, we give next a direct way to compute it that avoids the computation of lower envelopes.
In particular, we show that it is also possible to compute $\mu_c(e,e')$ in constant time by considering the volume of a three-dimensional body with a rectangular base 
(one side with the length of $e$ and the other with the length of $e'$), four vertical faces from each of the four base edges, and a \emph{roof} that is the lower envelope defined in Lemma \ref{lem:two_edges}. Next, we describe how this lower envelope or \emph{roof} can be viewed.

We consider a rectangle whose corners are labeled with the possible combinations of endpoints of $e=uv$ and $e'=u'v'$, as done to define the four planes in the proof of Lemma \ref{lem:two_edges}. The labels also include a weight equal to the distance between the corresponding endpoints; when no confusion may arise, we shall only indicate in the figures the weights of the corners. In addition, the rectangle is split (into at most four regions) by the orthogonal projection onto the $(x,y)$-plane of the (at most five) intersections of the planes defined by the equations in (\ref{planes}). Refer to Figure~\ref{fig:roof}.

Thus, for instance, a pair $(p,q)$, with $p\in e$ and $q\in e'$, is located in the region associated to $(u,u')$ if $d(p,q)$ is given by a path that goes through $u$ and $u'$, that is, by the plane  $z=|e|x+|e'|y+d(u,u')$ of Lemma \ref{lem:two_edges}. Hence, for $p=xv+(1-x)u$ and $q=yv'+(1-y)u'$ with $x,y\in [0,1]$, the value $d(p,q)$ is just the distance in the $L_1$ metric\footnote{The $L_1$-distance between two points $p=(x_p, y_p)$ and $q=(x_q, y_q)$ is given by $|x_p-x_q|+|y_p-y_q|$.} from $(p,q)$ to the corner $(u,u')$ plus the weight of that corner, which is $d(u,u')$. This is analogous for the remaining corners of the rectangle. Thus, the projections of the intersections between the planes of equations in (\ref{planes}) can be viewed as the bisectors of the additively weighted Voronoi diagram for the $L_1$ metric~\cite{okabe} of the corners of the rectangle. Therefore, the first step to compute $d(p,q)$ is to determine the region in which $(p,q)$ lies, as it determines the plane that defines the lower envelope over $(p,q)$.





The mean distance of the points in each of the (at most) four Voronoi regions is the volume of a truncated prism (with the corresponding Voronoi region as base and the corresponding plane of  Lemma \ref{lem:two_edges} as roof) divided by the area of the base. From a practical point of view, since there is no formula for a direct computation of that volume, it is better to subdivide the original Voronoi diagram into sub-rectangles and triangles, as Figure~\ref{fig:roof_subdiv} shows, since, in those cases, the volume of the truncated prism is given by the average height of the corners, see \cite{Klamkin1968}. In the proof of Proposition \ref{Kn}, we use this technique to compute the continuous mean distance of the $1$-uniform complete graph, and it is also referred in the proof of Proposition \ref{prop:bounds}.

\begin{figure}
     \centering
     \begin{subfigure}[b]{0.45\textwidth}
         \centering
         \includegraphics[page=1]{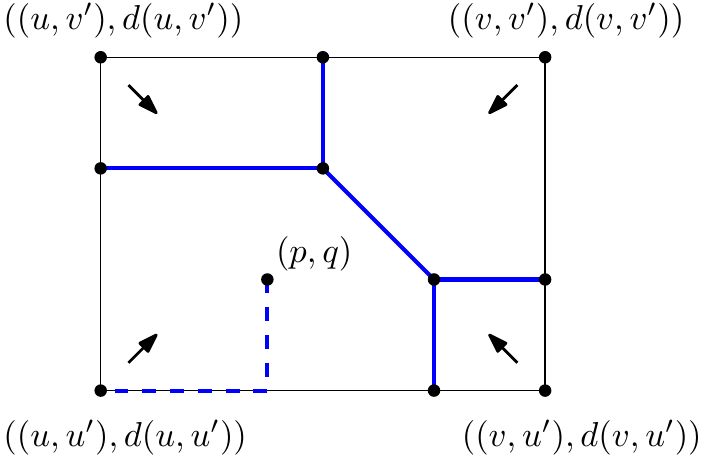}
         \caption{The length of the dashed path allows to compute the distance $d(p,q)$.}
         \label{fig:roof}
     \end{subfigure}
     \hfill
     \begin{subfigure}[b]{0.45\textwidth}
         \centering
         \includegraphics[page=2]{fig-roof23-1.pdf}
          \caption{Subdivision into simpler shapes for simpler computation.}
        \label{fig:roof_subdiv}
     \end{subfigure}
        \caption{The \emph{roof} determined by the distances between points on $e=uv$ and $e'=u'v'$, and a pair of points $(p,q)$ located in the region associated to $(u,u')$. Each corner is labeled with the height of the vertical edge of that corner, and the arrows show the direction of maximum slope of each roof.}
        \label{fig:roofs}
\end{figure}

\

\begin{remark}
Although we have already seen in Equation~(\ref{eq:path_formula}) that $\mu_c(e,e)=|e|/3$, we can give an interpretation of this value in terms of a roof-diagram (see Figure~\ref{fig:roof4}).
In this case, the rectangle becomes a square, and the roof is formed by two planes:
\begin{equation*} \label{one-edge}
    \begin{array}{cl}
    P(u,v):  &  z=|e|(y-x) \\
    P(v,u):  & z=|e|(x-y), \\
    \end{array}
\end{equation*}
\noindent where $P(u,v)$ indicates that point $p=xv+(1-x)u$ is closer to $u$ than $q=yv+(1-y)u$ (analogous for $P(v,u)$). 
\end{remark}

\begin{figure}[t]
\centering
\includegraphics{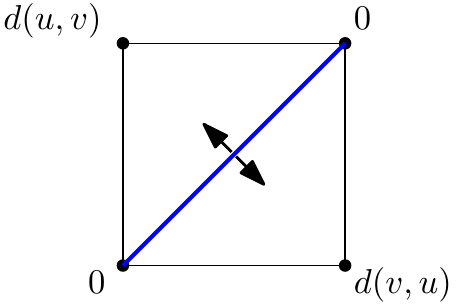}
\caption{The case $e=e'$.}
\label{fig:roof4}
\end{figure}

Summarizing the above discussion, we have the following theorem.

\begin{theorem}  \label{cor:floor}
Let $G$ be a connected weighted graph. Given two edges $e,e'\in E(\NL)$, the function $\mu_c(e,e')$ can be expressed as a weighted volume of at most eight truncated rectangular prisms. 
\end{theorem}

\section{Specific cases: trees, cactus, and complete graphs}
\label{sec:particular_cases}

The study developed in the previous section reflects the difficulties of computing the continuous mean distance, even for specific weighted graphs. 
As mentioned in the Introduction, the value of this parameter is only known for seven other simple graphs and six very specific graph families~\cite{DG82srmd,DG82mds}. In this section we deal with complete graphs and graphs that have cut vertices. For graphs that have this structural property, the continuous mean distance can be computed faster than using Theorem \ref{th:quadratic}, by studying each block independently.  

\begin{lemma} \label{lem:cutpoint}
Let $\NL$ be a weighted connected graph with a cut-vertex $v$, i.e., $\NL= \NL^1 \cup \NL^2$ and $\NL^1 \cap \NL^2 =\{v\}$. Then,
$\displaystyle\mu_c(\NL^1,\NL^2)=\mu_c(v,\NL^1)+\mu_c(v,\NL^2)$ and
$$\mu_c(\NL)=\left( \displaystyle{\frac{|\NL^1|}{|\NL|}} \right)^2 \mu_c(\NL^1) + \left( \displaystyle{\frac{|\NL^2|}{|\NL|}} \right)^2\mu_c(\NL^2)+2\left( \displaystyle{\frac{|\NL^1||\NL^2|}{|\NL|^2}} \right) \left(\mu_c(v,\NL^1)+\mu_c(v,\NL^2)\right).$$
\end{lemma}
\begin{proof} If $p \in \NL^1$ and $q \in \NL^2$, then $d(p,q)=d(p,v)+d(v,q)$. By Equations (\ref{eq:def1}) and (\ref{eq:def0}), we have: 

\begin{align*}
    \mu_c(\NL^1,\NL^2) &\stackrel{(\ref{eq:def1})}{=} \frac{1}{|\NL^1||\NL^2|}\iint_{p\in \NL^1, \, q \in \NL^2} d(p,q) \,dp\,dq \\
    &= \frac{1}{|\NL^1||\NL^2|}\left(\iint_{p\in \NL^1, \, q \in \NL^2} d(p,v)\,dp\,dq + \iint_{p\in \NL^1, \, q \in \NL^2} d(v,q)\,dp\,dq \right) \\
    &\stackrel{(\ref{eq:def0})}{=} \frac{1}{|\NL^1||\NL^2|}\left(|\NL^1|\int_{q \in \NL^2} \mu_c(v,\NL^1)dq + |\NL^2|\int_{p\in \NL^1} \mu_c(v,\NL^2) dp \right).
    \end{align*} 
Since in the last equation the integrated functions are constant with respect to the corresponding differentials, we obtain:
\begin{equation*}
    \mu_c(\NL^1,\NL^2) = \frac{1}{|\NL^1||\NL^2|}\left(\mu_c(v,\NL^1)|\NL^1||\NL^2|+\mu_c(v,\NL^2)|\NL^1||\NL^2| \right) 
    = \mu_c(v,\NL^1)+\mu_c(v,\NL^2).
\end{equation*} 
The formula for $\mu_c(\NL)$ is then obtained as follows:
\begin{align*}
    \mu_c(\NL) &\stackrel{(\ref{eq:def1})}{=} \frac{1}{|\NL|^2}\iint_{p,q\in \NL} d(p,q) \,dp\,dq \\
    &= \frac{1}{|\NL|^2}\left(\iint_{p,q\in \NL^1} d(p,q)\,dp\,dq + \iint_{p,q\in \NL^2} d(p,q)\,dp\,dq +2 \iint_{p\in \NL^1, \, q \in \NL^2} d(p,q)\,dp\,dq \right) \\
    & = \frac{1}{|\NL|^2}\left(|\NL^1|^2 \mu_c(\NL^1) + |\NL^2|^2 \mu_c(\NL^2)+2|\NL^1||\NL^2|  \mu_c(\NL^1,\NL^2) \right)  \\
    & = \frac{1}{|\NL|^2}\left(|\NL^1|^2 \mu_c(\NL^1) + |\NL^2|^2 \mu_c(\NL^2)+2|\NL^1||\NL^2|  (\mu_c(v,\NL^1)+\mu_c(v,\NL^2)) \right).
    \end{align*} 
\end{proof}

Thus, if we know a direct formula to obtain the continuous mean distance of each block of $\NL$,  then $\mu_c(\NL)$ can be computed in linear time.
For instance, this is the case for trees.

\begin{proposition}\label{prop:trees}
The continuous mean distance of a weighted tree $T_{\ell}$ with $n$ vertices can be computed in $O(n)$ time.
\end{proposition}
\begin{proof}
We apply induction on $n$. The continuous mean distance of an edge $e$ is, by Equation~(\ref{eq:path_formula}), $|e|/3$. For $n\geq 3$, take a non-leaf vertex $v$ of a tree $T_{\ell}$, which is a cut-vertex, and consider the two sub-trees connected by $v$. By Lemma \ref{lem:cutpoint},  $\mu_c(T_{\ell})$ is obtained by computing the total length of each sub-tree, its continuous mean distance, and its continuous mean distance from $v$. By induction, these values can be computed in linear time and combined in constant time to obtain $\mu_c(T_{\ell})$.
\end{proof}

 Another interesting application of Lemma \ref{lem:cutpoint} is for the well-known  \emph{cactus} graphs, see for instance \cite{FLLW-00, ZZ05} for studies in the context of location on graphs. This type of graphs has cut vertices, and each block is either an edge or a cycle. 
 Since the continuous mean distance of a cycle $C_{\ell}$ is $|C_{\ell}|/4$ (see \cite{DG82srmd}), and that of an edge is given by Equation~(\ref{eq:path_formula}), we obtain (again, by induction) the following result.
 
\begin{proposition}\label{prop:cactus}
The continuous mean distance of a weighted cactus graph with $n$ vertices can be computed in $O(n)$ time.
\end{proposition}
 
 When the graphs have no cut vertices, the method described in Subsection \ref{subsec:roof} is a useful tool to compute the continuous mean distance. Next, we apply this method to the $\alpha$-uniform complete graph $K_n^{\alpha}$. 
 While the value $\mu_d(K_n)=(n-1)/n$  is trivial to compute, the continuous version is much harder.

\begin{proposition}\label{Kn}
The continuous mean distance of the $\alpha$-uniform complete graph $K_n^{\alpha}$ is given by the following formula:
$$\mu_c ((K_n^{\alpha})_{\ell})=\frac{\alpha (9 \, n^{2} - 22 \, n + 12)}{6 \, {\left(n^{2} - n\right)}}.$$
\end{proposition}

\begin{proof}
By Remark \ref{rem:homotecia}, it suffices to prove the result for $\alpha=1$. We use Equation~(\ref{eq:sum}) and the technique described in Subsection \ref{subsec:roof} to compute the continuous mean distance between two distinct edges. There are two types of pairs of distinct edges $e$ and $e'$, incident and non-incident:

\,

\noindent \emph{Case 1.} If $e$ and $e'$ are incident at a vertex $u$, there is another edge connecting their non-common endpoints, say $v$ and $v'$. With respect to the description in Subsection \ref{subsec:roof}, we only have two planes:  $z=x+y$ (where $d(u,u')=0$) and  $z=2-x-y+d(v,v')$, and the corresponding roof--diagram is illustrated in Figure~\ref{fig:kn1}a. The partition of the diagram into one rectangle, one square and two triangles is shown in Figure~\ref{fig:kn1}b, where the number inside each region indicates the value of the volume of the truncated prism with that base; this number is given by the average height of the corners. Taking into account the corresponding areas of the base, we have $\mu_c(e,e')=\frac{1}{2} \cdot \frac{3}{4}+\frac{1}{4} \cdot \frac{4}{3}+\frac{1}{4}=\frac{23}{24}$, and the total number of this type of pairs of edges is $\displaystyle n(n-1)(n-2)$.

\,

\noindent \emph{Case 2.} If $e$ and $e'$ are non-incident, the roof--diagram looks as that of Figure~\ref{fig:kn2}. Now, $\mu_c(e,e')=3/2$, and the total number of these pairs of edges is 
$$\displaystyle {n\choose 2}^2-n(n-1)(n-2)-{n\choose 2},$$
where the last term is subtracted to take into account all pairs where $e=e'$, which are already considered in Equation~(\ref{eq:sum}). For $\alpha=1$, this equation then gives
$$\mu_c((K_n^{1})_{\ell})= \frac{4}{n^2(n-1)^2}\left( \sum_{(e,e')\in E \times E, e \neq e'}\mu_c(e,e')+\frac{n(n-1)}{6}\right)$$
where
$$\sum_{\substack{(e,e')\in E \times E \\ e \neq e}}\mu_c(e,e')=
\frac{23n(n-1)(n-2)}{24}+\frac{3}{2}\left(\frac{n^2(n-1)^2}{4}-n(n-1)(n-2)-\frac{n(n-1)}{2}\right).$$
Hence, 
$$\mu_c((K_n^{1})_{\ell})=\frac{92(n-2)}{24n(n-1)}+\frac{3}{2}-\frac{6(n-2)}{n(n-1)}-\frac{3}{n(n-1)}+\frac{2}{3n(n-1)}=\frac{9n^2-22n+12}{6n(n-1)}.$$
\end{proof}


\begin{figure}[ht]
\centering
\begin{tabular}{ccccccc}
\includegraphics{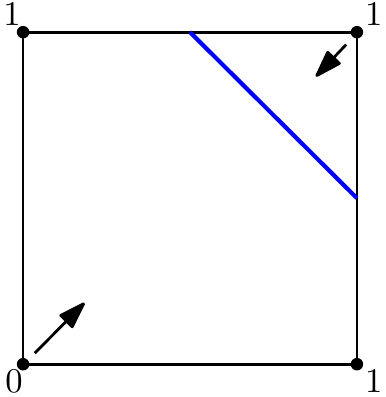}
&  &  &  & & &
\includegraphics{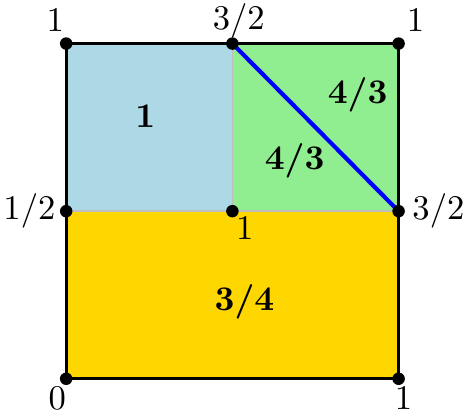}\\
(a)& & & & & & (b)
\end{tabular}
\caption{(a) Roof--diagram for two incident edges in $K_n^1$. 
 (b) Partition into truncated prisms; the value of the mean distance is indicated in each region in boldface, as well as the weights of the corners.}
\label{fig:kn1}
\end{figure}

\begin{figure}[ht]
\centering
\includegraphics{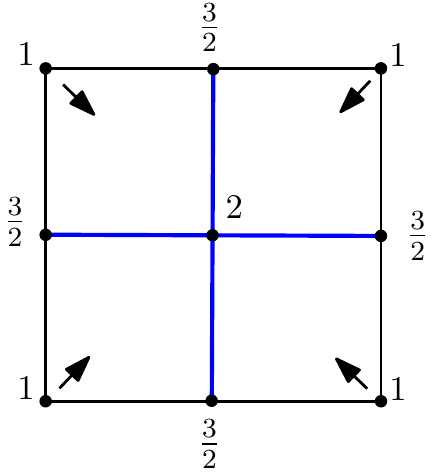}
\caption{The roof--diagram for two non-incident edges in $K_n^1$.}
\label{fig:kn2}
\end{figure}

Although this section is focused on the computation of the continuous mean distance of some specific weighted graphs, we conclude it with a result on the range of values of $\mu_c(T_{\ell})$, as it extends a similar result for the discrete case, which we believe is of interest. Indeed, in \cite{G-97}, the author proves that the Wiener index of any tree $T$ on $n$ vertices is lower-bounded by the Wiener index of the star $S$ on $n$ vertices, and upper-bounded by the Wiener index of the path $P$ with the same number of vertices (where the three graphs are unweighted). Therefore, by definition, $\mu_d(S)\leq \mu_d(T)\leq \mu_d(P)$. 
Next, we prove that, when the graph is uniform, these bounds also hold for the continuous case.    

\begin{proposition}
Let $S_{\ell}$ and $P_{\ell}$ be an $\alpha$-uniform star and $\alpha$-uniform path, respectively, on $n$ vertices. Then,
$$\mu_c(S_{\ell})\leq \mu_c(T_{\ell})\leq \mu_c(P_{\ell})$$
for every $\alpha$-uniform tree $T_{\ell}$ with $n$ vertices.
\end{proposition}

\begin{proof}
It suffices to prove the result for $\alpha=1$ (see Remark \ref{rem:homotecia}). We apply induction on $n$. Let $T_{\ell}':= T_{\ell}\setminus{\{u_1\}}$, $S_{\ell}':= S_{\ell}\setminus{\{u_2\}}$, and $P_{\ell}':= P_{\ell}\setminus{\{u_3\}}$ where $u_i$ is, in each case, a leaf adjacent to a vertex $v_i$, $1\leq i\leq 3$, of the corresponding graph. 
Lemma \ref{lem:cutpoint} gives expressions for $\mu_c(T_{\ell})$, $\mu_c(S_{\ell})$, and $\mu_c(P_{\ell})$ in terms of the continuous mean distances of the corresponding edge $u_iv_i$ and, respectively, $T_{\ell}'$, $S_{\ell}'$, and  $P_{\ell}'$ (simply set $\N_{\ell}^1$ as the  graph on $n-1$ vertices and $\N_{\ell}^2$ as the edge $u_iv_i$). For $T_{\ell}$ we obtain:
$$\mu_c(T_{\ell})=\left( \displaystyle{\frac{n-2}{n-1}} \right)^2 \mu_c(T_{\ell}') + \displaystyle{\frac{1}{3(n-1)^2}}+2\left( \displaystyle{\frac{n-2}{(n-1)^2}} \right) \left( \mu_c(v_1,T_{\ell}')+\frac{1}{2}\right),$$
and analogous expressions are obtained for $\mu_c(S_{\ell})$ and $\mu_c(P_{\ell})$ (by simply replacing $T_{\ell}'$ and $v_1$ by either $S'_{\ell}, v_2$ or $P'_{\ell}, v_3$, respectively). Hence, by induction, it suffices to prove that $\mu_c(v_2,S_{\ell}')\leq\mu_c(v_1,T_{\ell}')\leq\mu_c(v_3,P_{\ell}')$ where $v_2 $ is the central vertex of $S'_{\ell}$, and $v_3$ is an endpoint of $P'_{\ell}$. 

Given a vertex $v$ of any $1$-uniform tree $T$ with $m$ edges, by Equation (\ref{eq:def0}) and Lemma \ref{lem:vertex}, 
$$\mu_c(v,T)=\frac{\sum_{e \in T}\mu_c(v,e)}{m}=\frac{\frac{m}{2}+\sum_{e \in T}d(v,e)}{m},$$ since all edges in the tree belong to $E(T_v)$ and have weight $1$. Therefore, $\mu_c(v,T)$ is determined by the vector $V(v,T)=(d(v,e_1), d(v,e_2), \ldots , d(v,e_{m}))$, where $E(T)=\{e_1, \ldots, e_m\}$ and $d(v,e_i)$ is sorted in increasing order. The first coordinate of the vector is always 0 ($v$ belongs to at least one edge), and the difference between two consecutive coefficients of the vector is at most $1$ (the length of any edge in the tree). Thus, the smallest possible vector in a tree with $n-2$ edges (ordered by the sum of its coordinates) is $(0,0,\ldots,0)$, which corresponds to the case of the star, and the largest possible vector is $(0,1,2,\ldots,n-3)$, which corresponds to the path. This implies that $\mu_c(v_2,S_{\ell}')\leq\mu_c(v_1,T_{\ell}')\leq\mu_c(v_3,P_{\ell}')$.
\end{proof}


\section{Discrete versus continuous mean distances}
\label{sec:relation}

There is no obvious relation between the discrete and the continuous mean distances, in the sense that for different graphs, any of them can be larger.
From the result for $K_n^{\alpha}$ in the previous section (Proposition \ref{Kn}) it follows that the continuous mean distance can be larger than the discrete counterpart. This also happens for cycles\footnote{The continuous mean distance of an 1-uniform cycle $C_n$ of $n$ vertices is $n/4$  \cite{DG82srmd}, as well as $\mu_d(C_n)$ for $n$ even; otherwise $\mu_d(C_n)=\frac{n}{4}-\frac{1}{4n}$\cite{ mean_dist_M}.} and, on the other hand, we have seen (in the Introduction) that the opposite occurs for paths. 

This and the following section are devoted to better understanding the relationship between the two parameters. We first present bounds on the continuous mean distance of two edges in terms of discrete distances, which lead to bounds for the whole graph (in Corollary \ref{cor:bounds_same length} below) whenever it is uniform. 

\begin{proposition} \label{prop:bounds}
Let $e$ and $e'$ be two distinct edges in a weighted  graph $\N$. Then,
$$d(e, e')+\frac{|e|+|e'|}{4} \leq \mu_c(e, e') \leq d(e, e')+\frac{|e|+|e'|}{2},$$
and both bounds are tight.
\end{proposition}
\begin{proof}
For the upper bound, let $p\in e=ab$ and $q\in e'=uv$ and, without loss of generality, let $d(e,e')=d(a,u)$. We have $d(p,q)\leq d(p,a)+d(a,u)+d(u,q)= d(p,a)+d(e,e')+d(u,q)$. Hence, 
\begin{align*}
\mu_c(e,e') &\stackrel{(\ref{eq:def1})}{=}\frac{1}{|e||e'|}\iint_{p\in e, \, q \in e'} d(p,q) \,dp\,dq \\
            &\leq \frac{1}{|e||e'|}\left(\iint_{p\in e, \, q \in e'} d(p,a)\,dp\,dq + \iint_{p\in e, \, q \in e'} d(e,e')\,dp\,dq +\iint_{p\in e, \, q \in e'} d(u,q)\,dp\,dq \right) \\
   & = \frac{1}{|e||e'|}\left(\int_{q \in e'} \frac{|e|^2}{2} \,dq + d(e,e')|e||e'| +\int_{p\in e} \frac{|e'|^2}{2} \,dp \right) = \frac{|e|}{2}+ d(e, e')+\frac{|e'|}{2}.
    \end{align*}
Note that  $\int_{p\in e} d(p,a) \,dp=\frac{|e|^2}{2}$ since $a$ is an endpoint of the edge $e$, and so the integral is the area of a triangle with base and height equal to $|e|$. Analogously, $\int_{q\in e'} d(u,q) \,dq=\frac{|e'|^2}{2}$.

For the lower bound we have:
$$d(p,q)=\min\left\{ \begin{array}{c}
d(p,a)+d(q,u)+d(a,u) \\
d(p,b)+d(q,u)+d(b,u)\\
d(p,a)+d(q,v)+d(a,v)\\
d(p,b)+d(q,v)+d(b,v)
\end{array}
\right\} \geq \min\left\{ \begin{array}{c}
d(p,a)+d(q,u) \\
d(p,b)+d(q,u)\\
d(p,a)+d(q,v)\\
d(p,b)+d(q,v)
\end{array}
\right\} +d(e,e').$$
If the last minimum is denoted by $\Theta(p,q)$, then $d(p,q)\geq \Theta(p,q)+d(e,e')$. Therefore,
\begin{align*}
\mu_c(e,e')&\stackrel{(\ref{eq:def1})}{=}\frac{1}{|e||e'|} \iint_{p\in e, \, q \in e'} d(p,q) \, dp \, dq \\ & \geq \frac{1}{|e||e'|} \left( \iint_{p\in e, \, q \in e'} \Theta(p,q) \, dp \, dq+ \iint_{p\in e, \, q \in e'} d(e,e') \, dp \, dq \right) \\
& = \frac{1}{|e||e'|}\left(|e||e'|\frac{|e|+|e'|}{4}+|e||e'|d(e,e')\right) =\frac{|e|+|e'|}{4}+d(e,e'),
\end{align*}

\noindent where $\iint_{p\in e, \, q \in e'} \Theta(p,q) \, dp \, dq$ is the volume determined by the roof--diagram depicted in   Figure~\ref{fig:bound1}, which can be computed as follows:
\begin{align*}
\iint_{p\in e, \, q \in e'} \Theta(p,q) \, dp \, dq & =4 \iint_{p\in [a, \frac{a+b}{2}], \, q\in [u, \frac{u+v}{2}]} d(p,a)+d(q,u) \, dp \, dq \\ & = 4\left(\frac{|e|^2}{8}\cdot\frac{|e'|}{2}+\frac{|e'|^2}{8}\cdot\frac{|e|}{2}\right) =|e||e'|\left(\frac{|e|+|e'|}{4}\right).
\end{align*}
 
 \begin{figure}[ht]
\centering
\includegraphics{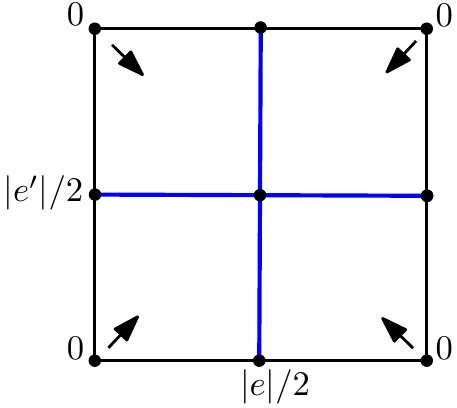}
\caption{The roof--diagram for $\Theta(p,q)$.}
\label{fig:bound1}
\end{figure}

Next, we observe that both bounds are tight. For instance, the mean distance of two edges that are connected by a unique path gives the upper bound, and the lower bound is attained by two edges whose endpoints are at the same distance so that $d(e,e')$ is given by any of the four possible combinations of endpoints.
\end{proof}

As a consequence of the preceding proposition we obtain, for $\alpha$-uniform graphs $\N$, bounds on $\mu_c(\NL)$  in terms of the discrete mean distance of a weighted version of its line graph. Recall that the \emph{line graph} $L(\N)$ of an unweighted graph $\N$ has a vertex associated with each edge in $\N$, and two vertices are adjacent if the corresponding edges of $\N$ have a vertex in common. When $\N$ is $\alpha$-uniform, we consider  the \emph{$\alpha$-uniform} line graph $L_{\alpha}(\N)$ that is defined analogously but, in addition, every edge has length $\alpha$.

The Wiener index of $L(\N)$ is known as the \emph{edge-Wiener index} of $\N$ (see, for instance, \cite{dgms-09,z-19} and the references therein).
It is defined as $\W_e(\N)=\sum_{\{e,e'\} \subseteq E} d_{L(\N)}(e,e')$,  where $d_{L(\N)}(e,e')$ is the distance of the corresponding vertices (to $e$ and $e'$) in $L(\N)$. This  can be naturally extended to $L_{\alpha}(\N)$ for $\alpha$-uniform graphs $\N$ with $m$ edges, and thus we may consider the discrete mean distance $\mu_d(L_{\alpha}(\N))=2\W(L_{\alpha}(\N))/m^2$.



\begin{corollary}\label{cor:bounds_same length}
 Let $\N$ be an $\alpha$-uniform graph with $m$ edges. Then,
$$\mu_d(L_{\alpha}(\N))+\frac{\alpha}{3m}-\frac{(m-1)\alpha}{2m} \leq \mu_c(\NL) \leq \mu_d(L_{\alpha}(\N))+\frac{\alpha}{3m}$$
where $L_{\alpha}(\N)$ is the $\alpha$-uniform line graph of $\N$.
\end{corollary}

\begin{proof}
It can be easily checked that, by construction, $d(e,e')=d_{L(\N)}(e,e')-1$ for distinct edges $e$ and $e'$ in an unweighted graph $\N$; this extends to $d(e,e')=d_{L_{\alpha}(\N)}(e,e')-\alpha$ when the graph $\N$ is $\alpha$ uniform. Proposition \ref{prop:bounds} then gives $d_{L_{\alpha}(\N)}(e,e')-\frac{\alpha}{2} \leq \mu_c(e, e') \leq d_{L_{\alpha}(\N)}(e,e')$ for distinct edges $e$ and $e'$ of $\NL$.
Now, by Equation (\ref{eq:sum}), we obtain: $$\mu_c(\NL) = \frac{1}{\alpha^2m^2}\left( \sum_{(e,e')\in E \times E, e \neq e'}{\mu_c(e,e')  \alpha^2 + \sum_{e \in E}{\frac{\alpha^3}{3}}} \right)=\frac{1}{m^2} \sum_{(e,e')\in E \times E, e \neq e'}\mu_c(e,e') +\frac{\alpha}{3m}.$$
Therefore,
$$\mu_c(\NL)\geq \frac{1}{m^2}\left(\sum_{(e,e')\in E \times E, e \neq e'}d_{L_{\alpha}(\N)}(e,e')\right) +\frac{\alpha}{3m}-\frac{(m-1)\alpha}{2m}$$
and
$$
\mu_c(\NL)\leq \frac{1}{m^2} \left( \sum_{(e,e')\in E \times E, e \neq e'} d_{L_{\alpha}(\N)}(e,e') \right) +\frac{\alpha}{3m}.
$$
The result then follows by definition of $\mu_d(L_{\alpha}(\N))$.
\end{proof}

In their seminal work on the mean distance for shapes~\cite{DG82mds}, Doyle and Graver already studied the continuous mean distance by using an iterative edge refinement process.
Following this direction, next we present an edge subdivision approach for trees. One may think that if the subdivision points are chosen  arbitrarily, the discrete mean distance of the refined tree may increase significantly  with respect to the original one. 
However, this is not the case, as we can see in the next theorem.




\begin{theorem}\label{th:convergence trees}
Let $T$ be a weighted tree with $n$ vertices, and let $T^{(k)}$ be the tree resulting from subdividing each edge of $T$ by adding $k$ new vertices on it.
Then $\mu_d(T^{(k)}) <\displaystyle \frac{n}{n-\frac{2k}{k+1}}\mu_d(T)$.
\end{theorem}

\begin{proof}
We start by noting that $T^{(k)}$ has $n+k(n-1)$ vertices: $n$ \emph{old} vertices and $k(n-1)$ \emph{new} vertices. The value
$\mu_d(T^{(k)})$ is the average of $(n+k(n-1))^2$ distances, of three types: 
(i) between two old vertices,
(ii) between an old and a new vertex,
(iii) between two new vertices.
To avoid any confusion, we shall use $d_T$ and $d_{T^{(k)}}$ to indicate distances in, respectively, the trees $T$ and $T^{(k)}$. Further, with some abuse of notation, we shall write $u\in T$ instead of $u\in V(T)$.

Consider a distance of type (ii), between an old vertex $u \in T$ and a new vertex $a \in T^{(k)} \setminus T$.
Since $T^{(k)}$ is a tree, there is a unique path from $u$ to $a$.
Moreover, since $a$ is interior to an edge of $T$, the path can be extended in direction away from $u$ until the first old vertex $v \in T$.
Clearly, $d_{T^{(k)}}(u,a) \leq d_{T^{(k)}}(u,v)=d_T(u,v)$.

Similarly, associated to a distance of type (iii), between two new vertices $a,b \in T^{(k)} \setminus T$, there is a unique path in $T^{(k)}$ that can be extended in both directions until starting and ending, respectively, at vertices $u,v \in T$.
Again, $d_{T^{(k)}}(a,b)\leq d_{T^{(k)}}(u,v)=d_T(u,v)$.

In this way, each distance involving a new vertex (types ii or iii) can be upper-bounded by a distance between two old vertices.
Moreover, the distance $d(u,v)$ of a pair of vertices $(u,v) \in T^2$ can only be an upper bound for up to $(k^2+2k)$ distances: $2k$ of type (ii), and $k^2$ of type (iii). Observe that the same happens for the distance $d(v,u)$.
This leads to an upper-bound on $\mu_d(T^{(k)})$ as follows:

\small{\[
\mu_d(T^{(k)})  =  \frac{\displaystyle
\sum_{(u,v) \in T^2} d_{T^{(k)}}(u,v) + \!\!\!\!\!\!\!\!
\sum_{(u,a) \in T\times (T^{(k)} \setminus T)} \!\!\!\!\!\!\!\! d_{T^{(k)}}(u,a)  +\!\!\!\!\!\!\!\! \sum_{(a,u) \in (T^{(k)}  \setminus T)\times T} \!\!\!\!\!\!\!\! d_{T^{(k)}}(a,u)+ \!\!\!\!\!\!\!\!
\sum_{(a,b) \in (T^{(k)} \setminus T)^2} \!\!\!\!\!\!\!\! d_{T^{(k)}}(a,b)}
{(n(k+1)-k)^2}\]}

\[\leq \frac{\displaystyle
\sum_{(u,v) \in T^2} d_{T}(u,v) +
(k^2+2k)\sum_{(u,v) \in T} d_{T}(u,v) 
}
{
(n(k+1)-k)^2
}
= \frac{\displaystyle (k+1)^2 \sum_{(u,v) \in T^2} d_{T}(u,v)}{(n(k+1)-k)^2}.\]

\,

Since $(n(k+1)-k)^2=n^2(k+1)^2(1-\frac{2k}{n(k+1)})+k^2>n^2(k+1)^2(1-\frac{2k}{n(k+1)})$, we obtain:
\[
\mu_d(T^{(k)})  < \frac{(k+1)^2 \displaystyle \sum_{(u,v) \in T^2} d_{T}(u,v)}{n^2(k+1)^2(1-\frac{2k}{n(k+1)})} =  \frac{\mu_d(T)}{1-\frac{2k}{n(k+1)}}=\frac{n}{n-\frac{2k}{k+1}}\mu_d(T).
\] 
\end{proof}

Theorem \ref{th:convergence trees} gives an upper bound of $\frac{n}{n-2}\mu_d(T)$ for the discrete mean distance of any subdivision of a tree $T$ but, if the subdivision points  are chosen more carefully, one can expect more precise results. Thus, Theorem \ref{th:convergence trees} is the initial motivation of the following section where, in particular, we explore the convergence of the discrete mean distance to its continuous counterpart when subdividing the edges of a tree (see Corollary \ref{cor:trees}).





\section{Convergence: graph subdivision}\label{sec:convergence_subdivisions}

A natural question is whether the discrete mean distance is convergent to its continuous counterpart when iteratively subdividing the edges of the graph. One may propose different subdivision schemes but, as we shall see later in this section, not all of them guarantee convergence. By definition of continuous mean distance, the convergence happens for uniform graphs by simply adding, at each step, a new vertex on each edge. 
Further, it is not hard to devise a subdivision scheme with guaranteed convergence if the ratio between the lengths of the longest and the shortest edges approaches one as the subdivision progresses.
However, such a scheme completely depends on the original structure of the graph. 

In this section we present an edge subdivision scheme that does not depend on the graph structure, and allows us to obtain bounds on the discrete mean distance of its $k$-th edge subdivision, and on its limit when $k$ tends to infinity.  We begin by introducing some notation. 

For a graph $\N=(V,E)$ with $n$ vertices and $m$ edges, let $\N^1=(V^1, E^1)$ be the graph that results from subdividing each edge of $\N$ by inserting a new vertex at its midpoint. Then, for a given $k\geq 2$, we subdivide each edge of $\N^1$ into $2^{k-1}$ new edges of the same length by inserting $2^{k-1}-1$ vertices.
The resulting graph $\N^k=(V^k,E^k)$ is called the \emph{$k$-th subdivision of $\N$} (note that the graph $G^1$ could be viewed as a subdivision of $G$ but for our purpose it will be distinguished). Refer to Figure \ref{fig:subdivision(1)}. The vertices of the original graph $\N$ are called \emph{black vertices}; the set of vertices inserted into $\N$ to obtain $\N^1$ is denoted by $\mathcal{B}$, and they are called \emph{blue vertices}. Thus, $V^1=V\cup \mathcal{B}$ and $|\mathcal{B}|=m$. We use $\mathcal{R}^k$ to refer to the set of new vertices inserted into $\N^1$, which are called the \emph{red vertices}; clearly, $|\mathcal{R}^k|=2m(2^{k-1}-1)=m(2^k-2)$. Hence, $V^k=V\cup \mathcal{B} \cup \mathcal{R}^k$ and $|V^k|=n+m(2^k-1)$. Further, the edges of the original graph $\N$ can be identified within $\N^k$: we write $e^k$ to indicate the $k$-th subdivision of an edge $e\in E$, which is a path in $G^k$ with $2^k+1$ vertices (there are 2 black, 1 blue, and $2^k-2$ red vertices).
 Equation (\ref{eq:discrete}) then gives:
\begin{equation}\label{discrete-wiener}
\mu_d(\N^k)=\frac{2\W(\N^k)}{(n+m(2^k-1))^2},
\end{equation}
where $\W(\N^k)=\sum_{\{u,v\}\subset V^k} d(u,v)$ is the Wiener index of $G^k$. With some abuse of notation, we write, for sets $A,B\subseteq V^k$, $\W(A;B)=\sum_{u\in A\setminus B, \, v\in B\setminus A} d(u,v)+\sum_{\{u,v\}\subseteq A\cap B} d(u,v)$ and $\W(A)=\W(A;A)$. Further, we shall indistinctly use sets of vertices or graphs in this notation, for instance, $\W(A;e^k)$ is simply a sum of distances between vertices in the set $A$ and vertices in the path $e^k$. 
With this notation, for $k\geq 2$, we have:
\begin{equation}\label{eq:red}
\W(\N^k)=\W(V^1)+\W(\mathcal{R}^k)+\W(\mathcal{R}^k; V)+\W(\mathcal{R}^k; \mathcal{B})
\end{equation}

\begin{figure}[t]
\centering
\includegraphics[width=0.95\textwidth]{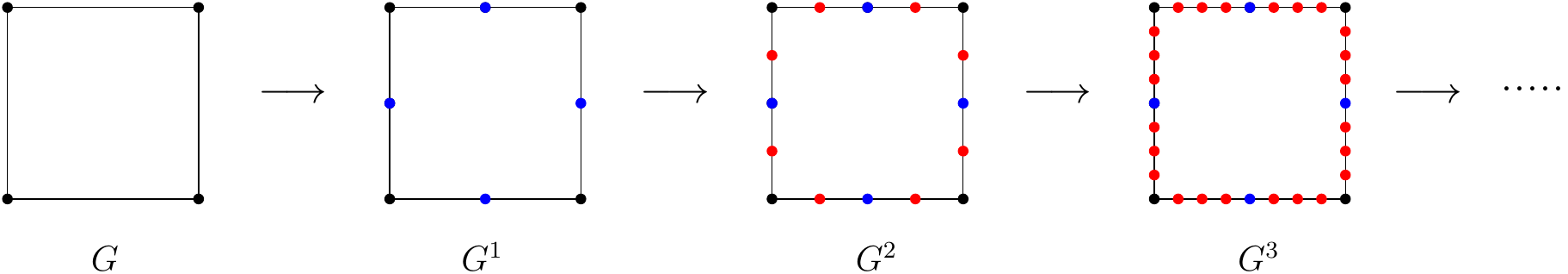}\caption{Subdividing the edges of a graph.}
\label{fig:subdivision(1)}
\end{figure}

The following two subsections are devoted to proving our main result in this section (which, in addition, will allow us to gain a deeper understanding of the limit of $\mu_d(\N^k)$ when $k$ tends to infinity):

\begin{theorem}\label{th:upper_lower_bounds} Let $\N=(V,E)$ be a weigthed graph with $n$ vertices and $m\geq 2$ edges, and let $\N^k=(V^k, E^k)$ be its $k$-th subdivision,  where $k\geq 2$. Let $\mathcal{B}$ be the set of vertices inserted into $\N$ to obtain the graph $\N^1$. Then,
$$\frac{2 \, \left[\Omega_k(\N, \N^1)-\rho \left(3{m\choose 2} +m(n-2)\right)\left(2^{k-2}-\frac{1}{2}\right)\right]}{(n+m(2^k-1))^2} < \mu_d(\N^k)\leq \frac{2 \, \Omega_k(\N, \N^1)}{(n+m(2^k-1))^2}$$
where $\rho={\rm max}\{|e| \, : \, e\in E \}$, and $$\Omega_k(\N, \N^1)=\W(V^1)+(2^k-2)\left(2^k\W(\mathcal{B})+\W(\mathcal{B};V) \right)+|E|\left(\frac{2^{2k-1}}{3}- 2^{k-1}+\frac{1}{3}\right).$$ Moreover, the upper bound is tight.
\end{theorem}

The limit of the upper bound in Theorem \ref{th:upper_lower_bounds}, when $k$ tends to infinity, is given by the coefficients of the term $2^{2k}$. Thus,
$$\lim_{k\to\infty} \mu_d(\N^k)\leq \lim_{k\to\infty} \frac{2 \, \Omega_k(\N, \N^1)}{(n+m(2^k-1))^2} = \frac{2(\W(\mathcal{B})+|E|/6)}{m^2}.$$
Hence, by Equation (\ref{eq:discrete}), we obtain the following bounds.

\begin{corollary}
Let $\N=(V,E)$ be a weighted graph with $m\geq 2$ edges, and let $\N^k=(V^k, E^k)$ be its $k$-th subdivision, where $k\geq 2$. Let $\mathcal{B}$ be the set of vertices inserted into $\N$ to obtain the graph $\N^1$. Then,
$$\lim_{k\to\infty} \mu_d(\N^k) \leq \mu_d(\mathcal{B})+\frac{|E|}{3m^2}.$$
Moreover, if $\N$ is $\alpha$-uniform then,  
$$
\lim_{k\to\infty} \mu_d(\N^k) = \mu_c(\NL) \leq\mu_d(\mathcal{B})+\displaystyle\frac{\alpha}{3m}.$$
\end{corollary}

We want to highlight that all trees attain the preceding upper bounds (this is a consequence of the study developed in Section \ref{sec:upper} below). 

\begin{corollary}\label{cor:trees}
Let $T=(V,E)$ be a weighted tree with $n\geq 3$ vertices, and let $T^k$ be its $k$-th subdivision. Let $\mathcal{B}$ be the set of vertices inserted into $T$ to obtain $T^1$. Then,
$$\lim_{k\to\infty} \mu_d(T^k) = \mu_d(\mathcal{B})+\frac{|E|}{3(n-1)^2}.$$
Moreover, if $T$ is $\alpha$-uniform then,
$$\mu_c(T_{\ell})=\lim_{k\to\infty} \mu_d(T^k)=\mu_d(\mathcal{B})+\displaystyle\frac{\alpha}{3(n-1)}.$$
\end{corollary}

Corollary~\ref{cor:trees} gives simple examples where the discrete mean distance, when subdividing the edges of the graph, does not converge to the continuous counterpart.  Consider, for example, a path $P$ with 4 vertices and edge lengths $2,1,1$; we have $\mu_c(\PP)=4/3 \approx 1.33 $, whilst from Corollary~\ref{cor:trees} we obtain $\lim_{k\to\infty} \mu_d(P^k)=10/9+4/27=34/27 \approx 1.26$.

\subsection{Proof of the upper bound in Theorem \ref{th:upper_lower_bounds}}\label{sec:upper}

First, we upper bound the Wiener index of $\N^k$. By Equation (\ref{eq:red}), it suffices to compute upper bounds on $\W(\mathcal{R}^k), \W(\mathcal{R}^k; V)$, and $\W(\mathcal{R}^k; \mathcal{B})$.

To upper-bound $\W(\mathcal{R}^k)$, we begin by distinguishing the distances between red vertices, depending on whether they are on the $k$-th subdivision of the same edge or of distinct edges. We use the notation $e^k\cap \mathcal{R}^k$ to indicate the set of red points that are on the $k$-th subdivision of an edge $e$; analogously, the notation $e^k\cap \mathcal{B}$ will refer to the blue point that is on $e^k$.
\begin{equation}\label{eq:same_distinct_edges}
\W(\mathcal{R}^k)=\sum_{\substack{\{e_1,e_2\} \subseteq E \\ e_1\neq e_2}}\W(e_1^k\cap \mathcal{R}^k; e_2^k\cap \mathcal{R}^k) + \sum_{e\in E} \W(e^k\cap \mathcal{R}^k).
\end{equation}
For red vertices that are on the same edge, we can compute the sum of distances exactly.

\begin{lemma} \label{lem:same_edge} The following formula holds for $k\geq 2$:
$$\sum_{e\in E} \W(e^k\cap \mathcal{R}^k)=|E|\left(\frac{2^{2k-1}}{3}-3\cdot 2^{k-2}+\frac{5}{6}\right).$$
\end{lemma}

\begin{proof}
Given an edge $e=wz\in E$, we have:
\begin{equation}\label{eq:total}
 \W(e^k\cap \mathcal{R}^k) = \W(e^k)-\W(\{w,z\};e^k)-\W(e^k\cap \mathcal{R}^k;e^k\cap \mathcal{B}).  
\end{equation}
The value $\W(e^k)$ is the Wiener index of an $\frac{|e|}{2^k}$-uniform path on $2^k+1$ vertices, which can be deduced from $\mu_d(e^k)$ by Equation (\ref{eq:discrete}) and Remark \ref{rem:homotecia}.
It is also known that the discrete mean distance of a 1-uniform path on $n$ vertices is $(n+1)(n-1)/3n$ \cite{mean_dist_M}. Thus, we have:
\begin{equation}\label{eq:path}
 \W(e^k)=\frac{|e|}{2^k}\left(\frac{\mu_d(e^k)(2^k+1)^2}{2}\right)=\frac{|e|}{2^k}\left(\frac{(2^k+2)(2^k+1)2^k}{6}\right)=\frac{|e|}{2^k}{2^k+2 \choose 3}.   
\end{equation}
Further, 
\begin{equation}\label{eq:black}
\W(\{w,z\};e^k)=d(w,z) + \mkern-18mu \sum_{\substack{u\in \{w,z\}\\v\in e^k-\{w,z\}}} \mkern-18mu d(u,v)  =|e|+2\left(\frac{|e|}{2^k}(1+2+3+\ldots+2^k-1)\right)=2^k|e|.
\end{equation}
\noindent Finally, 
\begin{equation}\label{eq:red-blue}
\W(e^k\cap \mathcal{R}^k;e^k\cap \mathcal{B})=2\left(\frac{|e|}{2^k}(1+2+3+\ldots +(2^{k-1}-1))\right)=|e|\left(2^{k-2}-\frac{1}{2}\right).\end{equation}
By Equations (\ref{eq:total}) -- (\ref{eq:red-blue}) we obtain:
$$\sum_{e\in E} \W(e^k\cap \mathcal{R}^k)=|E|\left(\frac{1}{2^k}{2^k+2 \choose 3}-2^k-2^{k-2}+\frac{1}{2}\right)=|E|\left(\frac{2^{2k-1}}{3}-3\cdot 2^{k-2}+\frac{5}{6}\right).$$
\end{proof}

It remains to upper-bound the sum of distances, for every pair of distinct edges, of red vertices that are on the subdivision of the edges; see Equation (\ref{eq:same_distinct_edges}). Roughly speaking, the maximum value of this sum is obtained when, for every pair of distinct edges, all  shortest path between any two of their vertices use the same endpoints of those edges; see Figure~\ref{fig:subdivisions}a, which considers the subdivisions $e_1^3$ and $e_2^3$ of two distinct edges $e_1$ and $e_2$. 
This is because when we can enter and get out of the edges using different endpoints, the distance between any two vertices (black, red, or blue) on the subdivision of the edges, in the best case, decreases. For example, in Figure~\ref{fig:subdivisions}a, $d(r_6, s_5)$ would be smaller if there would be another shortest path (not only the orange one) connecting the two edges via the endpoints $z_1$ and $z_2$. The following lemma gives the maximum value of that sum of distances.

\begin{lemma}\label{lem:distinct_edges}
 If for every two distinct edges $e_1$ and $e_2$ all shortest paths connecting any two vertices located on $e_1^k$ and $e_2^k$, respectively, go through the same endpoints of $e_1$ and $e_2$, then:
$$\sum_{\{e_1,e_2\}\in E, e_1\neq e_2}\W(e_1^k\cap \mathcal{R}^k, e_2^k\cap \mathcal{R}^k)=(2^k-2)^2\W(\mathcal{B}).$$
\end{lemma}

\begin{proof}
Consider two distinct edges $e_1=w_1z_1$ and $e_2=w_2z_2$, and assume that all shortest paths connecting any two vertices on $e_1^k$ and $e_2^k$ go through $w_1$ and $w_2$ (the argument is analogous for the other combinations of endpoints); it might happen that $w_1=w_2$. Let $\{w_1, r_1, \ldots r_{2^{k-1}-1}, b_1,  r_{2^{k-1}}, \ldots r_{2^{k}-2}, z_1\}$ be the sequence of vertices in $e_1^k$ ordered from the leftmost to the rightmost vertex, where $b_1\in \mathcal{B}$ and $r_i\in \mathcal{R}^k$, and let $\{w_2, s_1, \ldots s_{2^{k-1}-1}, b_2, s_{2^{k-1}}, \linebreak \ldots s_{2^{k}-2}, z_2\}$ be the analogous sequence of vertices in $e_2^k$. Refer to Figure~\ref{fig:subdivisions}a. We have:
{\small 
$$d(r_i,s_j)=\left\{\begin{tabular}{lcl}
 $\displaystyle d(b_1,b_2)-(2^{k-1}-i)\frac{|e_1|}{2^k}-(2^{k-1}-j)\frac{|e_2|}{2^k}$  &  if & $1\leq i,j \leq 2^{k-1}-1$\\ \noalign{\medskip}
  $\displaystyle d(b_1,b_2)+(i-2^{k-1}+1)\frac{|e_1|}{2^k}+(j-2^{k-1}+1)\frac{|e_2|}{2^k}$  &  if & $2^{k-1}\leq i,j \leq 2^{k}-2$\\ \noalign{\medskip}
    $\displaystyle d(b_1,b_2)-(2^{k-1}-i)\frac{|e_1|}{2^k}+(j-2^{k-1}+1)\frac{|e_2|}{2^k}$  &  if &  $1\leq i \leq 2^{k-1}-1$, $2^{k-1}\leq j \leq 2^{k}-2$\\ \noalign{\medskip}
     $\displaystyle d(b_1,b_2)+(i-2^{k-1}+1)\frac{|e_1|}{2^k}-(2^{k-1}-j)\frac{|e_2|}{2^k}$  &  if &  $1\leq j \leq 2^{k-1}-1$, $2^{k-1}\leq i \leq 2^{k}-2$\\ 
\end{tabular}\right.$$}

\noindent Hence, $\W(e_1^k\cap \mathcal{R}^k, e_2^k\cap \mathcal{R}^k)=\sum_{r_i\in e_1^k \cap \mathcal{R}^k, s_j\in e_2^k \cap \mathcal{R}^k} d(r_i,s_j)=(2^{k}-2)^2d(b_1,b_2)$, since all the expressions depending on $\frac{|e_1|}{2^k}$ and $\frac{|e_2|}{2^k}$ cancel each other out (they cancel out in pairs, for example, the expressions in $d(r_1,s_1)$ cancel out with the ones in $d(r_{2^k-2},s_{2^k-2})$). The result then follows by summing over all pairs of distinct edges $e_1$ and $e_2$.
\end{proof}

\begin{figure}[t]
\centering
\includegraphics[width=0.9\textwidth]{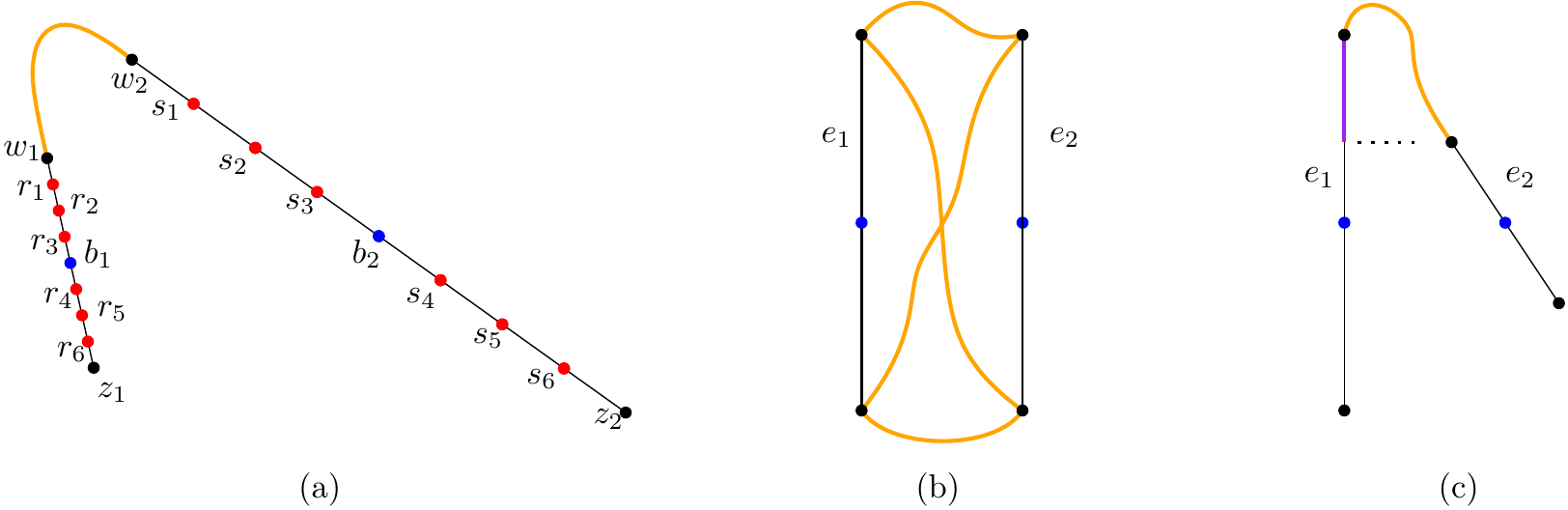}
\caption{(a) Subdivisions $e_1^3$ and $e_2^3$: the shortest path between any two vertices (black, red, or blue) located on, respectively, each subdivision, goes through the orange path. (b) The orange paths are shortest paths connecting the blue points (midpoints of $e_1$ and $e_2$). (c) The shortest path between any point on the sub-edge of $e_1$ in purple uses the orange path, i.e., the same endpoints of $e_1$ and $e_2$; this is the condition used in Lemma \ref{lem:distinct_edges} that increases the value of the sum of distances between vertices, respectively, on $e_1^k$ and $e_2^k$.}
\label{fig:subdivisions}
\end{figure}

As explained before, by Equation (\ref{eq:same_distinct_edges}), and Lemmas \ref{lem:same_edge} and \ref{lem:distinct_edges}, we obtain an upper bound on $\W(\mathcal{R}^k)$ for the $k$-th subdivision of any weighted graph (with at least two edges):
\begin{equation}\label{rojos}
\W(\mathcal{R}^k)\leq (2^k-2)^2\W(\mathcal{B})+|E|\left(\frac{2^{2k-1}}{3}-3\cdot 2^{k-2}+\frac{5}{6}\right).
\end{equation}

An upper bound on $\W(\mathcal{R}^k; V)$ is obtained by using similar arguments as in Lemma \ref{lem:distinct_edges}. Here we assume 
that for every edge $e=wz$ of $\N$, all shortest paths connecting any vertex on $e^k$ with any vertex in $V\setminus\{w, z\}$ go through the same endpoint $w$ of $e$. The only difference with the proof of Lemma \ref{lem:distinct_edges} is that we compute $d(r_i,v)$ for $v\in V\setminus\{w, z\}$ instead of $d(r_i,s_j)$, obtaining analogous expressions but distinguishing only the cases $1\leq i\leq 2^{k-1}-1$ and $2^{k-1}\leq i \leq 2^{k}-2$. The same type of expression is obtained for the endpoints of the edge $e$.
For example, $d(r_i, w)$ is given by:

$$d(r_i,w)=\left\{\begin{tabular}{lcl}
 $\displaystyle d(b,w)-(2^{k-1}-i)\frac{|e|}{2^k}$  &  if & $1\leq i \leq 2^{k-1}-1$\\ \noalign{\medskip}
  $\displaystyle d(b,w)+(i-2^{k-1}+1)\frac{|e|}{2^k}$  &  if & $2^{k-1}\leq i \leq 2^{k}-2$\\ \noalign{\medskip}
\end{tabular}\right.$$

\noindent where $b\in e^k \cap\mathcal{B}$. Thus, for a fixed $v\in V$ it follows that
$\sum_{r_i\in e^k \cap \mathcal{R}^k} d(r_i,v)=(2^{k}-2)d(b,v)$ (all the expressions depending on $\frac{|e|}{2^k}$ again cancel each other out). Therefore,
\begin{equation}\label{eq:red-black}
\W(\mathcal{R}^k; V)=\sum_{r_i\in \mathcal{R}^k, v\in V} d(r_i,v)\leq (2^{k}-2)\W(\mathcal{B}; V).
\end{equation}
and the bound is attained when the condition on the shortest paths stated above holds.

 By distinguishing again between vertices that are on the same edge or on distinct edges, and proceeding as above, we reach the following upper bound on $\W(\mathcal{R}^k; \mathcal{B})$.
  For the sake of brevity, we omit the details as the arguments are the same.
\begin{equation}\label{eq:rojos-azules}
\W(\mathcal{R}^k; \mathcal{B})\leq 2(2^k-2)\W(\mathcal{B})+|E|\left(2^{k-2}-\frac{1}{2}\right).
\end{equation}
The preceding bound is attained when,  for every edge $e$ of $\N$, all shortest paths connecting any vertex on $e^k$ with any vertex in $\mathcal{B}\setminus\{b\}$ (where $b\in e^k$) go through the same endpoint of $e$.

By Equation (\ref{eq:red}), and the bounds given in (\ref{rojos}), (\ref{eq:red-black}), and (\ref{eq:rojos-azules}), we can conclude that $\W(\N^k)\leq \Omega_k(\N, \N^1)$; the upper bound in Theorem \ref{th:upper_lower_bounds} then follows by Equation (\ref{discrete-wiener}). This bound is attained by all weighted graphs satisfying the conditions on shortest paths that lead to the equality in Equations (\ref{rojos}-- \ref{eq:rojos-azules}), in particular all trees.

\subsection{Proof of the lower bound in Theorem \ref{th:upper_lower_bounds}}

 The minimum value of $\mu_d(\N^k)$ would be obtained by a graph $\N$ satisfying that every pair of edges $e_1, e_2$ have the same length, and their midpoints
 are connected by shortest paths going through any pair of endpoints of $e_1$ and $e_2$ (see Figure~\ref{fig:subdivisions}b). Indeed, as we explained for the upper bound, the sum of distances between vertices on the subdivisions of the edges (black, red, or blue) decreases when the combinations of endpoints to enter and get out of the edges increase, so the minimum is given when all possible combinations of endpoints can be used. In addition, the graph should be uniform, as otherwise there would be a pair of edges in the situation described in Figure~\ref{fig:subdivisions}c, which would give a larger value for the sum of distances. 
 
 Clearly, there cannot exist a graph satisfying the previous condition on the shortest paths connecting the midpoints of any pair of edges (simply consider the midpoints of two incident edges), but, in order to make the computations necessary to obtain a lower bound on $\mu_d(\N^k)$, we shall  assume in Lemma \ref{lem_lower:distinct_edges} below that all edges of the graph $\N$ have the same length $\alpha$, and that any pair of its edges satisfies the condition on the midpoints.


By Equation (\ref{eq:red}), it suffices to compute lower bounds on $\W(\mathcal{R}^k), \W(\mathcal{R}^k; V)$, and $\W(\mathcal{R}^k; \mathcal{B})$ in order to lower-bound the Wiener index of $\N^k$. Again, we begin with $\W(\mathcal{R}^k)$.

\begin{lemma}\label{lem_lower:distinct_edges}
Let $G$ be an $\alpha$-uniform graph with $m\geq 2$ edges. If the midpoints of every pair $e_1, e_2$ of distinct edges of $\N$ are connected by shortest paths going through any pair of endpoints of $e_1$ and $e_2$, then:
 $$\sum_{\{e_1,e_2\}\in E, e_1\neq e_2}\W(e_1^k\cap \mathcal{R}^k, e_2^k\cap \mathcal{R}^k)=(2^k-2)^2\W(\mathcal{B})- 2\alpha {m\choose 2} \left(2^{k-2}-\frac{1}{2}\right).$$
\end{lemma}

\begin{proof}
Let $e_1=w_1z_1$ and $e_2=w_2z_2$, and let $b_1,b_2 \in \mathcal{B} $ be their corresponding midpoints.   We follow the same notation as in the proof of Lemma \ref{lem:distinct_edges} where $\{w_1, r_1, \ldots r_{2^{k-1}-1}, b_1,  r_{2^{k-1}}, \ldots \linebreak r_{2^{k}-2}, z_1\}$ and $\{w_2, s_1, \ldots s_{2^{k-1}-1}, b_2, s_{2^{k-1}}, \ldots s_{2^{k}-2}, z_2\}$ are the ordered sequence of vertices in, respectively, $e_1^k$ and $e_2^k$. Thus,
{\small $$d(r_i,s_j)=\left\{\begin{tabular}{lcl}
 $\displaystyle d(b_1,b_2)-(2^{k-1}-i)\frac{\alpha}{2^k}-(2^{k-1}-j)\frac{\alpha}{2^k}$  &  if & $1\leq i,j \leq 2^{k-1}-1$\\ \noalign{\medskip}
  $\displaystyle d(b_1,b_2)-(i-2^{k-1}+1)\frac{\alpha}{2^k}-(j-2^{k-1}+1)\frac{\alpha}{2^k}$  &  if & $2^{k-1}\leq i,j \leq 2^{k}-2$\\ \noalign{\medskip}
    $\displaystyle d(b_1,b_2)-(2^{k-1}-i)\frac{\alpha}{2^k}-(j-2^{k-1}+1)\frac{\alpha}{2^k}$  &  if &  $1\leq i \leq 2^{k-1}-1$, $2^{k-1}\leq j \leq 2^{k}-2$\\ \noalign{\medskip}
     $\displaystyle d(b_1,b_2)-(i-2^{k-1}+1)\frac{\alpha}{2^k}-(2^{k-1}-j)\frac{\alpha}{2^k}$  &  if &  $1\leq j \leq 2^{k-1}-1$, $2^{k-1}\leq i \leq 2^{k}-2$\\ 
\end{tabular}\right.$$}
Hence, we obtain: $$\W(e_1^k\cap \mathcal{R}^k, e_2^k\cap \mathcal{R}^k)=(2^k-2)^2d(b_1,b_2)-\left(\W(e_1^k\cap \mathcal{R}^k, e_1^k\cap \mathcal{B})+ \W(e_2^k\cap \mathcal{R}^k, e_2^k\cap \mathcal{B})\right),$$
which by Equation (\ref{eq:red-blue}) equals $(2^k-2)^2d(b_1,b_2)-2\alpha (2^{k-2}-1/2)$. When considering all pairs of distinct edges, the desired formula is obtained.
\end{proof}

As it was explained before, there is no graph satisfying the conditions of Lemma \ref{lem_lower:distinct_edges}, but it yields, together with Equation (\ref{eq:same_distinct_edges}) and Lemma \ref{lem:same_edge},  a lower bound on $\W(\mathcal{R}^k)$ by considering $\rho={\rm max}\{|e| \, : \, e\in E \}$.
\begin{equation}\label{lower_reds}
\W(\mathcal{R}^k)\geq (2^k-2)^2\W(\mathcal{B})-2\rho {m\choose 2} \left(2^{k-2}-\frac{1}{2}\right)+|E|\left(\frac{2^{2k-1}}{3}-3\cdot 2^{k-2}+\frac{5}{6}\right).
\end{equation}

We apply similar arguments to bound $\W(\mathcal{R}^k; V)$ and $\W(\mathcal{R}^k; \mathcal{B})$. In both cases we distinguish whether the vertices are on the same edge or on distinct edges. For all pairs of distinct edges $e_1$ and $e_2$ (all edges of the same length $\alpha$), we also assume that their midpoints  are connected by shortest paths going through any pair of endpoints of $e_1$ and $e_2$. For $\W(\mathcal{R}^k; V)$ we have:
$$\sum_{r_i\in e_1^k \cap \mathcal{R}^k} d(r_i,v)=(2^{k}-2)d(b_1,v)-2\frac{\alpha}{2^k}(1+2+3+\ldots+(2^{k-1}-1))=(2^{k}-2)d(b,v)-\alpha (2^{k-2}-1/2)$$
where $v\in e_2^k\cap V$ and $b_1\in e_1^k\cap \mathcal{B}$. Further, $\W(e_1^k\cap \mathcal{R}^k; e_1^k\cap V)=\alpha (2^k-2)$.  By considering again $\rho={\rm max}\{|e| \, : \, e\in E \}$ as done for Equation (\ref{lower_reds}), we obtain:
\begin{equation}\label{eq:red_black_lower}
\W(\mathcal{R}^k; V)\geq (2^k-2)(\W(\mathcal{B}; V)-m\rho)-\rho m(n-2)\left(2^{k-2}-\frac{1}{2}\right)+m\rho(2^k-2).
\end{equation}

An analogous process gives $\W(e_1^k\cap \mathcal{R}^k, e_2^k\cap \mathcal{B})=(2^k-2)d(b_1,b_2)-\alpha (2^{k-2}-1/2)$, and together with Equation (\ref{eq:red-blue}) leads to:
\begin{equation}\label{eq:red_blue_lower}
\W(\mathcal{R}^k; \mathcal{B})\geq 2(2^k-2)\W(\mathcal{B})+\left[ |E|- \rho {m\choose 2} \right]\left(2^{k-2}-\frac{1}{2}\right).
\end{equation}

Equation (\ref{eq:red}), and the bounds in (\ref{lower_reds}), (\ref{eq:red_black_lower}), and (\ref{eq:red_blue_lower}) imply:

$$\W(\N^k)\geq\Omega_k(\N, \N^1)-\rho \left(3{m\choose 2} +m(n-2)\right)\left(2^{k-2}-\frac{1}{2}\right).$$ 

The lower bound in Theorem \ref{th:upper_lower_bounds} again follows from Equation (\ref{discrete-wiener}).

\section{Conclusions and future work}



In this work we have presented the first thorough study of the continuous mean distance, a natural graph parameter that has received little attention until now.
From a computational perspective, we presented two different methods to compute the mean distance of a weighted graph in roughly quadratic time in the number of edges.
In addition, we obtained several structural results that provide a deeper understanding of this parameter, and can also be used to compute the mean distance faster for several graph classes. 
Finally, we studied the relation between the discrete mean distance and the continuous counterpart, in order to understand how the iterative subdivision of edges makes the discrete mean distance converge to the continuous one.

We are left with many intriguing questions for future research.
The computational complexity of the continuous mean distance is far from settled.
An important question is for what other graph classes the continuous mean distance can be computed in subquadratic time.
In the case of the discrete mean distance, this was recently shown to be possible for planar graphs~\cite{Cabello2019}, so it is worth studying if similar techniques could be applied to the continuous setting.
If that is not possible, one can still resort to approximation algorithms. For this, it can be useful to understand further the relation between the discrete and the continuous mean distance, since for instance, proving a constant factor relation between them would lead to subquadratic approximation algorithms for planar graphs.


\subsection*{Acknowledgements} 
We are grateful to Julian Pfeifle for proposing the topic of this work, and for many stimulating discussions about it. We also thank the anonymous reviewer for an extremely detailed and constructive review, which has helped to improve the presentation of this work considerably.

\bibliographystyle{plainurl}
\bibliography{refs}

\section*{Statements and Declarations}
This work was supported by project PID2019-104129GB-I00/ AEI/ 10.13039/501100011033.
R.~S.\ was also supported by project Gen.\ Cat.\ 2017SGR1640.
A.M. was also supported by project BFU2016-74975-P.

\noindent The authors have no relevant financial or non-financial interests to disclose.

\noindent All authors contributed to the manuscript equally.

\paragraph{Data availability} 
 Data sharing not applicable to this article as no datasets were generated or analyzed during the current study.
 
 \paragraph{Conflicts of interest} 
No conflicts of interest reported for this work.

\end{document}